\documentclass[10pt,twocolumn]{article}

\usepackage{times}

\input{main.sty}

\allowdisplaybreaks

\begin{document}

\title{\bf Applying Quantitative Semantics\\
  to Higher-Order Quantum Computing\thanks{Partially 
    founded by French ANR project \textsc{Coquas} 
    (number 12 JS02 006 01) and CNRS chair ``Logique lin\'eaire et calcul''.}}

\author{\begin{tabular}{c}
    Michele Pagani\\
    \normalsize Universit\'e Paris 13, Sorbonne Paris Cit\'e\\
    \normalsize Villetaneuse, France\\
    \normalsize michele.pagani@lipn.univ-paris13.fr
  \end{tabular}
  \begin{tabular}{c}
    Peter Selinger\\
    \normalsize Dalhousie University\\ 
    \normalsize Halifax, Canada\\
    \normalsize selinger@mathstat.dal.ca
  \end{tabular}
  \begin{tabular}{c}
    Beno\^{i}t Valiron\\
    \normalsize CIS Dept, University of Pennsylvania\\ 
    \normalsize Philadelphia, U.S.A.\\
    \normalsize benoit.valiron@monoidal.net
  \end{tabular}
  }

\date{}

\maketitle

\begin{abstract}
 Finding a denotational semantics for higher order quantum
  computation is a long-standing problem in the semantics of quantum
  programming languages. Most past approaches to this problem fell
  short in one way or another, either limiting the language to an
  unusably small finitary fragment, or giving up important features of
  quantum physics such as entanglement. In this paper, we propose a
  denotational semantics for a quantum lambda calculus with recursion
  and an infinite data type, using constructions from quantitative
  semantics of linear logic.
\end{abstract}

\section{Introduction}

Type theory and denotational semantics have been successfully used to
model, design, and reason about programming languages for almost half
a century. The application of such methods to quantum computing is
much more recent, going back only about 10 years {\cite{Selinger04}}.

An important problem in the semantics of quantum computing is how to
combine quantum computing with higher-order functions, or in other
words, how to design a functional quantum programming language. A
syntactic answer to this question was arguably given with the design
of the quantum lambda calculus
{\cite{valiron08phd,selinger06lambda}}. The quantum lambda calculus
has a well-defined syntax and operational semantics, with a strong
type system and a practical type inference algorithm. However, the
question of how to give a {\em denotational} semantics to the quantum
lambda calculus turned out to be difficult, and has remained open for
many years {\cite{Selinger04b,SV09}}. One reason that designing such a
semantics is difficult is that quantum computation is inherently
defined on {\em finite dimensional} Hilbert spaces, whereas the
semantics of higher-order functional programming languages, including
such features as infinite data types and recursion, is inherently
infinitary.

In recent years, a number of solutions have been proposed to the
problem of finding a denotational semantics of higher-order quantum
computation, with varying degrees of success. The first
approach~\cite{valiron06fully} was to restrict the language to strict
linearity, meaning that each function had to use each argument exactly
once, in the spirit of linear logic. In this way, all infinitary
concepts (such as infinite types and recursion) were eliminated from
the language. Not surprisingly, the resulting finitary language
permitted a fully abstract semantics in terms of finite dimensional
spaces; this was hardly an acceptable solution to the general problem.
The second approach~\cite{malherbe2010} was to construct a semantics
of higher-order quantum computation by methods from category theory;
specifically, by applying a presheaf construction to a model of
first-order quantum computation. This indeed succeeds in yielding a
model of the full quantum lambda calculus, albeit without
recursion. The main drawbacks of the presheaf model are the absence of
recursion, and the fact that such models are relatively difficult to
reason about. The third
approach~\cite{GoIquantum} was based on the Geometry of
Interaction. Starting from a traced monoidal category of basic quantum
operations, Hasuo and Hoshino applied a sequence of categorical
constructions, which eventually yielded a model of higher-order
quantum computation. The problem with this approach is that the tensor
product constructed from the geometry-of-interaction construction does
not coincide with the tensor product of the underlying physical data
types. Therefore, the model drops the possibility of entangled states,
and thereby fails to model one of the defining features of quantum
computation. 

\noindent
{\bf Our contribution.}
\quad 
In this paper, we give a novel denotational semantics of higher-order
quantum computation, based on methods from {\em quantitative
  semantics}. Quantitative semantics refers to a family of semantics
of linear logic that interpret proofs as linear mappings between
vector spaces (or more generally, modules), and standard lambda terms
as power series.  The original idea comes from Girard's normal functor
semantics \cite{Girard88c}. More recently, quantitative semantics has
been used to give a solid, denotational semantics for various
algebraic extensions of lambda calculus, such as probabilistic and
differential lambda calculi (e.g.\ \cite{danosehrhard}, \cite{finsp}).

One feature of our model is that it can represent {\em infinite
  dimensional} structures, and is expressive enough to describe
recursive types, such as lists of qubits, and to model recursion. This
is achieved by providing an exponential structure {\em \`{a} la}
linear logic. Unlike the Hasuo-Hoshino model, our model permits
general entanglement. We interpret (a minor variant of) the quantum
lambda calculus in this model.  Our main result is the adequacy of the
model with respect to the operational semantics.

The model is the juxtaposition of a simple, finite-dimensional model
of quantum computation together with a canonical completion yielding
the structures of linear logic.  Our model demonstrates that the
quantum and the classical ``universes'' work well together, but also
-- surprisingly -- that they do not mix too much, even at higher order
types.

\noindent
{\bf Outline.}
\quad 
In Section~\ref{sec:background}, we briefly review some background.
Section~\ref{sec:qlc} presents the version of the quantum lambda
calculus that we use in this paper, including its operational
semantics.  Section~\ref{sec:sem} presents the denotational semantics
of the quantum lambda calculus, and Section~\ref{sec:adequacy} proves
the adequacy theorem. Section~\ref{subsect:discussion} concludes with
some properties of the representable elements.

\section{Background}\label{sec:background}

\subsection{Quantum computation in a nutshell}

Quantum computation is a computational paradigm based on the laws of
quantum physics. We briefly recall some basic notions; please see
{\cite{nielsen02quantum}} for a more complete treatment. The basic
unit of information in quantum computation is a {\em quantum bit} or
{\em qubit}, whose state is given by a normalized vector in the
two-dimensional Hilbert space $\C^2$. It is customary to write the
canonical basis of $\C^2$ as $\{\ket0,\ket1\}$, and to identify these
basis vectors with the booleans false and true, respectively. The
state of a qubit can therefore be thought of as a complex linear
combination $\alpha\ket0 + \beta\ket1$ of booleans, called a {\em
  quantum superposition}. More generally, the state of $n$ qubits is
an element of the $n$-fold tensor product
$\C^2\otimes\ldots\otimes\C^2$.

There are three kinds of basic operations on quantum data:
initializations, unitary maps and measurements. Initialization
prepares a new qubit in state $\ket0$ or $\ket1$. A unitary map, or
\define{gate}, is an invertible linear map $U$ such that $U^*=U\inv$;
here $U^*$ denotes the complex conjugate transpose of $U$. Finally,
the operation of measurement consumes a qubit and returns a classical
bit. If $n$ qubits are in state $\alpha\ket0\tensor\phi_0 +
\beta\ket1\tensor\phi_1$, where $\phi_0$ and $\phi_1$ are normalized
states of $n-1$ qubits, then measuring the leftmost qubit yields
 false with probability $|\alpha|^2$, leaving the remaining
qubits in state $\phi_0$, and true with probability $|\beta|^2$,
leaving the remaining qubits in state $\phi_1$.

\begin{example}\label{ex:cointoss}
  A small algorithm is the simulation of an unbiased coin toss:
  initialize one quantum bit to $\ket{0}$, apply the Hadamard gate
  sending $\ket0$ to $\frac1{\sqrt2}(\ket0+\ket1)$ and $\ket1$ to
  $\frac1{\sqrt2}(\ket0-\ket1)$, then measure. The result is
  true with probability $\frac12$ and false with probability
  $\frac12$.
\end{example}

\begin{figure}
\begin{center}
\scalebox{0.8}{$\xymatrix@=.7pc{
  &  &  &  &  &\ar@{}|*\txt{qubit 1:~~~~}  & \ket{\phi}\ar@{-}[r] & 
          *-={\bullet}\ar@{-}[r]\ar@{-}[dd] & 
          *+[F]{H}\ar@{-}[r] &
          *+<10pt>{\;}\\
  &  & \ar@{}|*!/lu8pt/\txt{(i)} &  &  &  && 
          \ar@{}|*!<15pt,5pt>\txt{(ii)}  &  & 
          *+{M}\ar@{}
          \ar@/^2.2pc/@{.>}[rrddddd]|!{[llllddd];
            [rrrddd]}\hole^<(0.2){x,y}\\
  \ar@{}|*\txt{qubit 2:~~~~}&\ket{0}\ar@{-}[r] & *+[F]{H}\ar@{-}[r] & 
          *-={\bullet}\ar@{-}[rrrr]\ar@{-}[dd] 
          & & & & *-={\oplus}\ar@{-}[rr] &  &
          *+<10pt>{\;}\\
  &  &  &  &  &  &  &  &  & *+<10pt>{\;}\\
  \ar@{}|*\txt{qubit 3:~~~~} & \ket{0}\ar@{-}[rr] &  & *-={\oplus}
          \ar@{-}`r[rd][rd]&&
          \ar@3{.}[rrrrrrrr]_(.6)*+\txt{location
            B}^(.4)*+\txt{location A}
          &&&&&&&&\\
  & & & &*-={}\ar@{-}`d[rd][rd]&& & & & & & & & & \\
  &  &  &  &  & *-={}\ar@{-}[rrrrrr] & & & & & & 
          *+[F]{U_{xy}}\ar@{-}[r] \ar@{}|*!<10pt,15pt>\txt{(iii)}
          & \ket{\phi}  
  \save
  "1,10"."3,10"*[F--]\frm{}
  \restore
  \save
  "3,3"."5,4"*!<3pt,-3ptpt>++[F.]\frm{}
  \restore
  \save
  "1,8"."3,10"*!<2pt,-6pt>++[F.]\frm{}
  \restore
}$}
\end{center}
\vspace{-10pt}
\caption{\footnotesize The quantum teleportation protocol.}\label{fig:telep}
\end{figure}
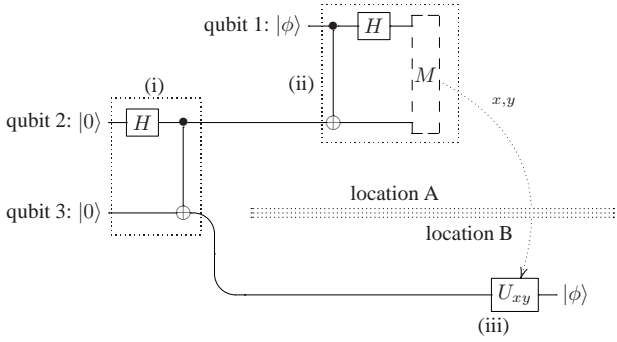
  
\begin{example}\label{ex:telep-presentation}
  A slightly more involved algorithm is the {\em quantum teleportation
    algorithm} (see {\cite{nielsen02quantum}} for details). The
  procedure is summarized in Figure~\ref{fig:telep}. Wires represent
  the path of quantum bits in the computation, and time flows from
  left to right. The gate \scalebox{0.8}{$\xy(0,0)*+[F]{H}\endxy$}
  stands for an application of the Hadamard gate, whereas the gate
  \raisebox{1ex}{\scalebox{0.8}{$\xy(0,0)*{\bullet};
      (0,-3)*{\oplus}**\dir{-}\endxy$}}
  is a controlled-not: it negates the bottom qubit if the upper one is
  in state $\ket1$. The box $M$ is a measurement. The unitaries
  $U_{xy}$ are
  \[
  U_{00} = \left(\begin{smallmatrix}1&0\\0&1\end{smallmatrix}\right),
  ~ U_{01} =
  \left(\begin{smallmatrix}0&1\\1&0\end{smallmatrix}\right), ~ U_{10}
  =
  \left(\begin{smallmatrix}1&0\\0&\textrm{-}1\end{smallmatrix}\right),
  ~ U_{11} =
  \left(\begin{smallmatrix}0&1\\\textrm{-}1&0\end{smallmatrix}\right).
  \]
  The goal is to send a quantum bit in an unknown state $\ket\phi$
  from Location A to Location B using two classical bits. The
  procedure can be reversed to send two classical bits using a quantum
  bit. In this case it is called the \define{dense coding algorithm}
  {\cite{nielsen02quantum}}.
  
  The algorithm consists of three parts. In (i), two quantum bits
  (qubits 2 and 3) are entangled in state
  $\frac1{\sqrt2}(\ket{00}+\ket{11})$. In (ii), the input qubit 1 in
  state $\ket\phi$ is entangled with qubit 2, then both are
  measured. The result is sent over location B, where in (iii) an
  correction $U_{xy}$ is applied on qubit 3, setting it to state
  $\ket\phi$.
\end{example}

\subsection{Density matrices and completely positive maps}

If we identify $\ket0$ and $\ket1$ with the standard basis vectors
$\zzmatrix{c}{1\\0}$ and $\zzmatrix{c}{0\\1}$, the
state of a qubit can be expressed as a two-dimensional vector
$v=\alpha\ket0 + \beta\ket1=\zzmatrix{c}{\alpha\\\beta}$. Similarly,
the state of an $n$-qubit system can be expressed as an
$2^n$-dimensional column vector.
Often, it is necessary to consider {\em probability distributions} on
quantum states; these are also known as {\em mixed states}. Consider a
quantum system that is in one of several states $v_1,\ldots,v_k$ with
probabilities $p_1,\ldots,p_k$, respectively. The {\em density matrix}
of this mixed state is defined to be
$ A = \sum_i p_i v_iv_i^*,
$
where $(-)^*$ denotes the adjoint operator.
By a theorem of Von Neumann, the density matrix is a good
representation of mixed states, in the following sense: two mixed
states are indistinguishable by any physical experiment if and only if
they have the same density matrix {\cite{nielsen02quantum}}. Note that
$\tr A = p_1+\ldots+p_k$. For our purposes, it is often convenient to
permit sub-probability distributions, so that $p_1+\ldots+p_k\leq 1$.

Let us write $\C^{n\times n}$ for the space of $n\times
n$-matrices. Recall that a matrix $A\in\C^{n\times n}$ is called {\em
  positive} if $v^*Av\geq 0$ for all $v\in\C^n$. Given
$A,B\in\C^{n\times n}$, we write $A\sqleq B$ iff $B-A$ is positive;
this is the so-called {\em {\Lowner} partial order}. A linear map
$F:\C^{n\times n}\to\C^{m\times m}$ is called {\em positive} if
$A\sqgeq 0$ implies $F(A)\sqgeq 0$, and {\em completely positive} if
$F\otimes\id_k$ is positive for all $k$, where $\id_k$ is the identity
function on $\C^{k\times k}$. If $F$ moreover satisfies
$\tr(F(A))\leq\tr A$ for all positive $A$, then it is called a {\em
  superoperator}. The density matrices are precisely the positive
matrices $A$ of trace $\leq 1$. Moreover, the superoperators
correspond precisely to those functions from mixed states to mixed
states that are physically possible {\cite{nielsen02quantum,Selinger04}}.

\subsection{The category \texorpdfstring{$\CPM$}{CPM}}

The category $\CPMs$ is defined as follows: the objects are natural
numbers, and a morphism $F:n\to m$ is a completely positive map
$F:\C^{n\times n}\to\C^{m\times m}$. Let $\CPM$ be the free completion
of $\CPMs$ under finite biproducts; specifically, the objects of
$\CPM$ are sequences $\vec n=(n_1,\ldots,n_k)$ of natural numbers,
and a morphism $F:\vec n\to\vec m$ is a matrix $(F_{ij})$ of morphisms
$F_{ij}:n_j\to m_i$ of $\CPMs$. The categories $\CPMs$ and $\CPM$ are
symmetric monoidal, and in fact, compact closed {\cite{Selinger04}}.

\subsection{Limitations of \texorpdfstring{$\CPM$}{CPM} as a model}
\label{sec:rec-type}

The category $\CPM$ can serve as a
fully abstract model for a simple, strictly linear, finitary quantum
lambda calculus {\cite{valiron06fully}}. For example, the type $\bit$
is interpreted as $(1,1)$, and the type $\qubit$ is interpreted as
$(2)$. Measurement, as a map from $\qubit$ to $\bit$, sends
$(\begin{smallmatrix}a&b\\c&d\end{smallmatrix})$ to $(a,d)$. The coin
toss is a map $(1)\to(1,1)$ sending $(p)$ to $(\frac{p}2,\frac{p}2)$.
Function spaces are interpreted via the compact closed structure.

As mentioned in the introduction, the semantics of
{\cite{valiron06fully}} is extremely limited, because it is completely
finitary. Thus recursion, infinite data types, and non-linear
functions (i.e., those that can use their argument more than once) had
to be completely removed from the language in order to fit the
model. For example, even the simple squaring function $f\mapsto
\lambda x.f(f\,x)$ is not representable in $\CPM$.

The purpose of the present paper is to remove all of these
restrictions. As an example, consider the following pseudo-code (in
ML-style):
\begin{Verbatim}[fontsize=\footnotesize,commandchars=\\\{\}]
val\textrm{\bf qlist} : qubit -> qubit list
let rec\textrm{\tt\bf qlist} q = if (cointoss) then [q]
            else let (x,y) = entangle q in x::(\textrm{\bf{}qlist} y)
\end{Verbatim}
Here, {\tt cointoss} is a fair coin toss, and the function
{\tt entangle} sends $\alpha\ket0+\beta\ket1$ to $\alpha\ket{00} +
\beta\ket{11}$.

So if the function ${\bf qlist}$ is applied to a qubit
$\alpha\ket0+\beta\ket1$, the output is $\alpha\ket0+\beta\ket1$ with
probability $\frac12$, $\alpha\ket{00}+\beta\ket{11}$ with probability
$\frac14$, $\alpha\ket{000}+\beta\ket{111}$ with probability
$\frac18$, and so on. Its semantics should be of type $2 \to
(2,4,8,\ldots)$, mapping
\[
 \left(\begin{smallmatrix}a&b\\c&d\end{smallmatrix}\right)\mapsto
 \left(
    \frac12\left(\begin{smallmatrix}a&b\\c&d\end{smallmatrix}\right),
    \frac14\left(\begin{smallmatrix}a&0&0&b\\0&0&0&0\\0&0&0&0\\
        c&0&0&d\end{smallmatrix}\right)
    ,\ldots
  \right).
\]
The category $\CPM$ is ``almost'' capable of handling this case, but
not quite, because it cannot express infinite tuples of matrices. The
model we propose in this paper is essentially an extension of $\CPM$
to infinite biproducts, using methods developed
in~\cite{Girard99coherentbanach,MelliesTT09,LairdMM12,LairdMMP13}.

\section{A quantum lambda calculus}\label{sec:qlc}

\begin{table}
\[
\begin{array}{l}
{\it Terms}\quad M,N,P
  \quad {:}{:}{=}
\\[1ex]
\qquad  x\bor \lambda x^A.M \bor MN\bor
        \punit\bor\letunitterm{M}{N}\bor\\[1ex]
\qquad  \tensterm{M}{N}\bor 
        \lettensterm{x^A}{y^B}{M}{N}\bor\\[1ex]
\qquad  \injl M\bor\injr M\bor\match P{x^A}M{y^B}N\bor\\[1ex]
\qquad  \splitlist[\!A]\bor
\letrec{f^{A{\multimap} B}}{\!x}{M\!}{\!N}
        \bor\meas \bor \new \bor U
\\[3ex]
{\it Values}\quad V,W
\quad {:}{:}{=}
\\[1ex]
\qquad  x \bor c \bor \lambda x^A.M \bor 
\tensterm{V}{W}\bor\injl{V}\bor\injr{W}
\\[3ex]
 {\it Types}\quad A, B, C
  \quad{:}{:}{=}
  \\[1ex]
\qquad  \qubit\bor A{\loli}B\bor \bang{(A\,{\loli}\, B)}\bor
  \tunit\bor A\,{\tensor}\,B \bor  A\,{\sumtype}\, B\bor\tlist{A}.
\end{array}
\]
\caption{\footnotesize Grammars of terms, values and types. 
  }\label{table:terms_grammar}
\end{table}

We define a variant of the typed quantum lambda calculus
of~\cite{SV09}.  The main difference is that the language in this
present paper is a true extension of linear logic (see the type
assignment system of Table~\ref{fig:typing rules}). In particular, in
contrast with~\cite{SV09}, $\oc( A\otimes B)\multimap\oc A\otimes \oc
B$ is not provable and there is no need for a subtyping relation.
The operational semantics implements a call-by-value strategy. An
untyped call-by-name variant has been studied in~\cite{LagoMZ11}.

The classes of \emph{terms}, \emph{values} and \emph{types} are
defined in Table~\ref{table:terms_grammar}. The symbol $c$ ranges over
the set of term constants $\{\punit, {\tt split}^A,$ $\meas, \new,
U\}$. The constant $U$ ranges over a set of elementary unitary
transformations on quantum bits. In the examples below, we will be
using the Hadamard gate $H$ and the controlled-not gate $N_c$, defined
as follows {\cite{nielsen02quantum}}:
\begin{align}
  H&=\frac{1}{\sqrt 2}
  	\left(\begin{smallmatrix}
   	 1&1\\
   	 1&-1
   	\end{smallmatrix}\right)
  &
  N_c&=
  	\left(\begin{smallmatrix}
   	 1&0&0&0\\
   	 0&1&0&0\\
	 0&0&0&1\\
	 0&0&1&0
   	\end{smallmatrix}
      \right)\label{eq:h_and_nc}
\end{align}  
 
Notice that bound variables
are given in Church style, i.e., with a type annotation. This enables
Proposition~\ref{prop:unicity_derivation}, and simplifies the semantic
interpretation of the typed terms. We omit such
annotations  in the sequel if uninteresting or obvious.

We have two kinds of arrows: the linear arrow $A{\multimap}B$, and the
intuitionistic arrow $\bang{(A\multimap B)}$, which is obtained by the
call-by-value translation of the intuitionistic implication into
linear logic \cite{ll}. Intuitively, only the terms of type
$\bang{(A\multimap B)}$ represent functions that can be used
repeatedly, whereas terms of type $A{\multimap}B$ must be used exactly
once.
A type of the form $\oc{A}$ is called a \emph{$\oc$-type} or
\emph{non-linear} type, and all other types are called {\em linear}.
The distinction between linear and non-linear types is crucial for
allowing the type system to enforce the no-cloning property of quantum
physics.

By convention, $\multimap$ is associative to the right, while
application and tensor are associative to the left. We use the
notation $A^{\tensor n}$ for $A$ tensored $n$ times.  The type
$\tlist{A}$ denotes finite lists of type $A$.  When doing structural
induction on types, we assume that $\tlist A$ is greater than
$A^{\otimes n}$, for any $n\in\N$.

The set of terms and types is somewhat spartan; however it
can be easily extended by introducing syntactic sugar. Note that,
for technical convenience, we have only allowed types of the form 
$\bang{A}$ when $A$ is an arrow type. However,
for an arbitrary type $A$, the type $\bang{A}$ can be simulated 
by using $\bang(\tunit\loli A)$ instead.
\begin{notation}\label{notation:syntax}
We write $\bit=\tunit\oplus\tunit$, 
$\ttrue=\injr\punit$,
$\ffalse=\injl\punit$,
$\nil=\injl\punit$ and
$\cons MN=\injr (M\otimes N)$. We
  write $\lambda\punit.M$ for the term $\lambda
  z^{\tunit}.(\letunitterm zM)$, where $z$ is a fresh variable, and
    $\iftermx PMN$ for $\match P{x^\tunit}{N}{y^\tunit}{M}$.
\end{notation}

\begin{table*}
\scalebox{0.9}{
\begin{minipage}{1.1\textwidth}
\[
\infer[ax]{\bang{\Delta},x:A\entail x:A}{\text{$A$ linear}}
\qquad
\infer[axd]{\bang{\Delta},x:\bang (A\loli B)\entail x:A\loli B}{}
\qquad
\infer[p]{\bang\Delta\entail V:\bang{(A\loli B)}}{
  \bang\Delta\entail V:A\loli B
  &
  V\textrm{ value}
}
\qquad
\infer[\tunit_I]{
  \bang{\Delta}\entail \punit:\tunit
}{}
\]

\[
\infer[\loli_I]{\Delta\entail\lambda x^A.M:A\loli B}{
  \Delta,x:A\entail M:B}
\quad
\infer[\loli_E]{\bang{\Delta},\Gamma,\Sigma\entail MN:B}{
  \bang\Delta,\Gamma\entail M:A\loli B
  &
  \bang\Delta,\Sigma\entail N:A
}
\quad
\infer[\tunit_E]{
  \bang\Delta,\Gamma,\Sigma\entail
  \letunitterm{M}{N}:A
}{
  \bang\Delta,\Gamma\entail M:\tunit
  &
  \bang\Delta,\Sigma\entail N:A
}
\]

\[
\infer[\tensor_I]{
  \bang{\Delta},\Gamma,\Sigma\entail \tensterm{M}{N}:A\tensor B
}{
  \bang\Delta,\Gamma\entail M:A
  &
  \bang\Delta,\Sigma\entail N:B
}
\quad
\infer[\tensor_E]{
  \bang{\Delta},\Gamma,\Sigma\entail 
  \lettensterm{x^{A}}{y^{B}}{M}{N}:C
}{
  \bang\Delta,\Gamma\entail M:A\tensor B
  &
  \bang\Delta,\Sigma,x:A,y:B\entail N:C
}
\]

\[
\infer[\oplus_{I}^\ell]{
  \bang\Delta,\Gamma\entail
  \injl{M}:A\sumtype B
}{
  \bang\Delta,\Gamma\entail M:A
}
\quad
\infer[\oplus_{I}^r]{
  \bang\Delta,\Gamma\entail
  \injr{M}:A\sumtype B
}{
  \bang\Delta,\Gamma\entail M:B
}
\quad
\infer[\oplus_E]{
  \bang{\Delta},\Gamma,\Sigma\entail \match{P}{x^A}{M}{y^B}N:C
}{
  \bang\Delta,\Gamma\entail P:A\sumtype B
  &
  \deduce{\bang\Delta,\Sigma,y:B\entail N:C}{\bang\Delta,\Sigma,x:A\entail M:C}
}
\]
\vspace{-1ex}
\[
\infer[\tlist{-\!}_{\!\!I}]{
  \bang\Delta,\Gamma\entail
  M:\tlist{A}
}{
  \bang\Delta,\Gamma\entail M:\tunit\oplus(A{\otimes}\tlist A)
}
\quad
\infer[\splitlist]{
  \bang\Delta\entail\splitlist[A] : 
  \tlist A {\multimap}\tunit\oplus(A{\otimes}\tlist A)
}{}
\quad
\infer[{\tt rec}]{
  \bang{\Delta},\Gamma\entail \letrec{f^{A\multimap B}}{x}{M}{N}:C
}{
  \bang\Delta,f:\bang{(A\multimap B)},x:A\entail M:B
  &
  \bang\Delta,\Gamma,f:\bang{(A\multimap B)}\entail N:C
}
\]

\[
\infer[\meas]{
  \bang\Delta\entail
  \meas:\qubit\multimap \bit
}{}
\quad
\infer[\new]{
  \bang\Delta\entail
  \new:\bit\multimap\qubit
}{}
\quad
\infer[U]{
  \bang\Delta\entail
  U:\qubit^{\otimes n}\multimap\qubit^{\otimes n}
}{U \text{ of arity $n$}}
\]
\end{minipage}
}
\caption{\footnotesize Typing rules. The contexts $\Gamma$ and 
  $\Sigma$ are assumed to be linear.}\label{fig:typing rules}
\end{table*}

A \emph{context} $\Delta$ is a function from a finite set of variables
to types. We denote the domain of $\Delta$ by $\supp\Delta$, and we
write $\Delta=x_1:A_1,\dots, x_n:A_n$ whenever
$\supp\Delta=\{x_1,\dots,x_n\}$ and $\Delta(x_i)=A_i$.  We call
$\Delta$ \emph{exponential} (resp.\ \emph{linear}) whenever all $A_i$
are $\oc$-types (resp.\ no $A_i$ is a $\oc$-type). We write
$\oc\Delta$ for a context that is exponential. The notation $\Gamma,
\Sigma$ refers to the union of the two contexts $\Gamma$ and $\Sigma$
and assumes that $\supp{\Gamma}$ and $\supp\Sigma$ are disjoint.

A \emph{judgement} is a triple $\Gamma\vdash M:A$ of a context
$\Gamma$, a term $M$ and a type $A$. A judgement is called \emph{valid}
if it can be inferred from the typing rules in
Figure~\ref{fig:typing rules}, using the convention that the contexts
$\Gamma$ and $\Sigma$ are linear.

\begin{proposition}\label{prop:unicity_derivation}
  There is at most one derivation inferring a given typing judgement
  $\Gamma\vdash M:A$.
  \qed
\end{proposition}

\begin{example}\label{ex:term_type}
  In Section~\ref{sec:rec-type}, we wrote the informal program ${\bf
    qlist}$. Our language is expressive enough to represent it. The
  term $\mathtt{cointoss}$ can be defined as $\meas (H (\new\,
  \ttrue))$, and it has type $\bit$. The term $\mathtt{entangle}$ is
  $\lambda x^\qubit.N_c(x\otimes(\new\,\ffalse))$, which has type
  $\qubit\multimap\qubit\otimes\qubit$. Then, ${\bf qlist}$ is
\begin{multline*}
\mathtt{letrec}\; f^{\qubit\multimap\qubit^\ell} q =\\ 
     \mathtt{if}\; \mathtt{cointoss}\; \mathtt{then}\; 
     \cons q\nil\hspace{2.5cm}\\
     \mathtt{else}\; \mathtt{let}\; x^\qubit\otimes y^\qubit 
     =\mathtt{entangle}\; q\; \mathtt{in} \;\cons {x}{fy}
\end{multline*}
which has type $\qubit\multimap\qubit^\ell$. In
Examples~\ref{ex:term_red} and~\ref{ex:term_sem} we discuss its
operational and denotational semantics, respectively.
\end{example}

\begin{example}\label{ex:telep-term}
  In Example~\ref{ex:telep-presentation} and Figure~\ref{fig:telep},
  we sketched the quantum teleportation algorithm. We said that the
  algorithm can be decomposed into 3 parts. Each of these parts can be
  described and typed in the quantum lambda calculus, yielding a
  higher-order term. This is an adaptation of an example provided
  in~\cite{selinger06lambda}.
  \begin{itemize}
  \item[(i)] generates an EPR pair of entangled quantum bits. Its type
    is therefore $\tunit\loli\qubit\tensor\qubit$. The corresponding
    term is 
    \[ 
    {\bf EPR} = \lambda\punit. N_c\left((H(\new\,\ffalse))\tensor(\new\,
      \ffalse)\right). 
      \]
  \item[(ii)] performs a Bell measurement on two quantum bits and
    outputs two classical bits $x,y$. Its type is thus
    $\qubit\loli\qubit\loli\bit\tensor\bit$, and the term {\bf
      BellMeasure} is defined as
    \[
     \lambda q_1.\lambda q_2.
     \left(\begin{array}{l}
         {\tt let}\ x\tensor y = N_c\,(q_1\tensor q_2)
         \\
         {\tt in}\ (\meas\,(H\,x))\tensor(\meas\,y)
       \end{array}\right).
   \]
  \item[(iii)] performs a correction. It takes one quantum bit, two
    classical bits, and outputs a quantum bit. It has a type of the form
    $\qubit\loli\bit\tensor\bit\loli\qubit$. The term is
    \[ \begin{array}{r}
     {\bf U} = \lambda q.\lambda x\tensor y.
     \mbox{\tt if\,$x$\,then\,$($if\,$y$\,then\,$U_{11}\,q$\,else\,$U_{10}\,q)$}
     \\ 
     \mbox{\tt
       else\,$($if\,$y$\,then\,$U_{01}\,q$\,else\,$U_{00}\,q)$}.
     \hspace{-.8ex}
   \end{array}
   \]
  \end{itemize}
  We can now write the term
  \[
  {\bf telep} = \begin{array}[t]{l}
    \begin{array}[t]{l@{}l@{~}l}
      \lambda\punit.&{\tt let}~{x\tensor y}&= {\bf EPR}~\punit~{\tt in}\\
      &{\tt let}\ {f} &= {\bf BellMeasure}\ x~{\tt in}\\
      &{\tt let}\ {g} &= {\bf U}\ y
    \end{array}\\
    \quad\qquad\hspace{1.3ex}{\tt in}\ {f\tensor g}.
  \end{array}
  \]
  It can then be shown that
  \[
  \entail {\bf telep}:!(\tunit\loli (\qubit \loli
  \bit\otimes\bit)\tensor(\bit\otimes\bit \loli \qubit))
  \]
  is a valid typing judgement.
  In other words, the teleportation algorithm produces a pair of
  entangled functions $f:\qubit\to\bit\tensor\bit$ and
  $g:\bit\tensor\bit\to\qubit$. These functions have the property that
  $g(f(\ket\phi))=\ket\phi$ for
  all qubits $\ket\phi$, and $f(g(x\tensor y))=(x\tensor y)$ for all
  booleans $x$ and $y$. These two functions are each other's inverse,
  but because they contain an embedded qubit each, they can only be
  used once. They can be said to form a ``single-use isomorphism''
  between the (otherwise non-isomorphic) types $\qubit$ and
  $\bit\tensor\bit$. However, the whole procedure is duplicable: one
  can generate as many one-time-use isomorphism pairs as desired.
\end{example}

\subsection{Operational semantics}
\label{sec:op-semant}

The operational semantics is defined in terms of an abstract machine
simulating the behavior of Knill's QRAM model~\cite{knill}. It is
similar to the semantics given in \cite{SV09}.

\begin{definition}
  \label{def:qclos}
  A \emph{quantum closure} is a triple $\am{\qarray,\qlist,M}$ where
  \begin{itemize}
  \item $\qarray$ is a normalized vector of $\C^{2^n}$, for some
    integer $n\geq 0$. The vector $\qarray$ is called the \emph{quantum
      state};
  \item $M$ is a term, not necessarily closed;
  \item $\qlist$ is a one-to-one map from the set of free variables of
    $M$ to the set $\{1,\ldots, n\}$. It is called the \emph{linking
      function}.
  \end{itemize}
  We write $\supp{\qlist}$ for the domain of $\qlist$.
  By abuse of language we may call a closure
  $\am{\qarray,\qlist,V}$ a \emph{value} when the term $V$ is a value.
  We denote the set of quantum closures by  $\Cl$ and the set of
  quantum closures that are values by $\Val$.
  We write $\ell|_M$ for the linking function whose domain is
  restricted to the set of free variables of $M$.
  We say that the quantum closure $\am{\qarray,\qlist,M}$ is
  \emph{total} when $|\qlist |$ has cardinality $n$, the size of the
  quantum state. In that case, if $|\qlist |=\{x_1,\ldots, x_n\}$ and
  $\qlist(x_i)=i$, we write $\qlist$ as $\ket{x_1,\ldots, x_n}$.  A
  quantum closure $\am{\qarray,\ket{x_1,\ldots, x_n},M}$ \emph{has a
    type $A$}, whenever $x_1:\qubit,\dots,x_n:\qubit\vdash M:A$. In
  case $\qlist=\ket{x_1,\ldots, x_n}$ we can also write $\qlist\vdash
  M:A$.
\end{definition}

The purpose of a quantum closure is to provide a mechanism to talk
about terms with embedded quantum data. The idea is that a variable
$y\in\FV(M)$ is bound in the closure $\am{\qarray,\qlist,M}$ to qubit
number $\qlist(y)$ of the quantum state $\qarray$. So for example,
the quantum closure
$
\am{\frac1{\sqrt2}(\ket{00}+\ket{11}),\ket{x_1,x_2},\lambda y^A.yx_1x_2}
$
denotes a term $\lambda y^A.yx_1x_2$ with two embedded qubits $x_1$, $x_2$ in
the entangled state $\ket{x_1x_2}=\frac1{\sqrt2}(\ket{00}+\ket{11})$. 

The notion of $\alpha$-equivalence extends naturally to quantum
closures, for instance, the states $\am{\qarray,\ket{x},\lambda
  y^A.x}$ and $\am{\qarray,\ket{z},\lambda y^A.z}$ are
equivalent. From now on, we tacitly
identify quantum closures up to renaming of bound variables.

The evaluation of a term is defined as a probabilistic rewriting
procedure on quantum closures, using a call-by-value reduction
strategy.  We use the notation $\am{\qarray,\qlist,
  M}\redto[p]\am{\qarray',\qlist', M'}$ to mean that the left-hand side
closure reduces in one step to the right-hand side with probability
$p\in[0,1]$.

\begin{table*}
  \footnotesize
  \centering
  \subfloat[Classical control.]{
    \label{subtable:reduction_classical}
    \centering
    \parbox{.95\textwidth}{
      \centering
      \begin{align*}
        \am{\qarray{},\qlist{},(\lambda x^A.M)\,V}
        &{}\redto[1]
        \am{\qarray{},\qlist{},M\{V/x\}}
        &
        \am{\qarray{},\qlist{},\lettensterm{x^A}{y^B}{V\tensor W}{N}}
        &{}\redto[1]
        \am{\qarray{},\qlist{},N\{V/x, W/y\}}
        \\
        \am{\qarray{},\qlist{},\letunitterm{\punit}{N}}
        &{}\redto[1]
        \am{\qarray{},\qlist{},N}
        &
        \am{\qarray{},\qlist{},\match{(\injl V)}{x^{\!A}}{M}{y^{\!B}}{N}}
        &{}\redto[1] \am{\qarray{},\qlist{},M\{V/x\}}        
        \\
        \am{\qarray{},\qlist{},\splitlist V}
        &{}\redto[1] \am{\qarray{},\qlist{},V}
        &
        \am{\qarray{},\qlist{},\match{(\injr V)}{x^{\!A}}{M}{y^{\!B}}{N}}
        &{}\redto[1] \am{\qarray{},\qlist{},N\{V/y\}}
      \end{align*}
      \vspace{-4ex}
      \begin{align*}
        \am{\qarray{},\qlist{},\letrec{f^{A\multimap B}}{x}{M}{N}}
        &{}\redto[1] \am{\qarray{},\qlist{},N\{(\lambda
          x^A.\letrec{f^{A\multimap B}}{x}{M}{M})/f\}}
      \end{align*}
    }
  }
  
  \subfloat[Quantum data. The variable $y$ is fresh. The decomposition
  of the quantum array in the case of $\meas\, x$ is explained in 
  Definition~\ref{def:rw}.]{
    \label{subtable:reduction_quantum}
    \centering
    \parbox{.95\textwidth}{
      \centering
      \begin{align*}
        \am{\qarray{},\qlist, U(x_{1}\otimes\dots\otimes x_{k})}
        &{}\redto[1]
        \am{\qarray{}',\qlist,x_{1}\otimes\dots\otimes x_{k}}
      \end{align*}
      \vspace{-4ex}
      \begin{align*}
        \am{\qarray,\emptyset,\new~\ffalse}
        &{}\redto[1]
        \am{\qarray\otimes\ket0,\{y\mapsto n+1\}, y}
        &
        \am{\alpha\qarray_0+\beta\qarray_1,\{x\mapsto i\},\meas~x}
        &{}\xredto[\abs{\beta}^2]
        \am{\qarray'_1,\emptyset, \ttrue}
        \\
        \am{\qarray,\emptyset,\new~\ttrue}
        &{}\redto[1]
        \am{\qarray\otimes\ket{1},\{y\mapsto n+1\}, y}
        &
        \am{\alpha\qarray_0+\beta\qarray_1,\{x\mapsto i\},\meas~x}
        &{}\xredto[\abs{\alpha}^2]
        \am{\qarray'_0,\emptyset, \ffalse}
      \end{align*}}
  }
  \\
  \subfloat[Congruence rules, under the hypothesis that for some
  $\ell_0$ we have $\ell=\ell_0 \uplus \ell|_M$, $\ell'=\ell_0 \uplus
  \ell'|_{M'}$ and
  $\am{\qarray,\qlist|_M,M}{\redto[p]}\am{\qarray',\qlist'|_{M'},M'}$.
  ]{\label{subtable:reduction_congruence}
    \centering
    \parbox{.95\textwidth}{
      \centering
      \begin{align*}
        \am{\qarray,\qlist{},MN}&\redto[p]\am{\qarray',\qlist{}',M'N}
        &
        \am{\qarray,\qlist{},M\otimes N}&
        \redto[p]\am{\qarray',\qlist{}',M'\otimes N}
        &
        \am{\qarray,\qlist{},\injl M}&
        \redto[p]\am{\qarray',\qlist{}',\injl M'}
        \\
        \am{\qarray,\qlist{},VM}&\redto[p]\am{\qarray',\qlist{}',VM'}
        &
        \am{\qarray,\qlist{},V\otimes M}&
        \redto[p]\am{\qarray',\qlist{}',V\otimes M'}
        &
        \am{\qarray,\qlist{},\injr M}&
        \redto[p]\am{\qarray{}',\qlist{}',\injr M'}
      \end{align*}
      \vspace{-4ex}
      \begin{align*}
        \am{\qarray,\qlist{},\letunitterm{M}{N}}&
        \redto[p]\am{\qarray',\qlist{}',\letunitterm{M'}{N}}
        &
        \am{\qarray,\qlist{},\lettensterm{x^A}{y^B}{M}{N}}&
        \redto[p]\am{\qarray',\qlist{}',\lettensterm{x^A}{y^B}{M'}{N}}
      \end{align*}
      \vspace{-4ex}
      \begin{align*}
        \am{\qarray,\qlist{},\match
          M{x^A}P{y^B}{N}}&\redto[p]\am{\qarray{}',
          \qlist{}',\match{M'}{x^A}P{y^B}{N}}
      \end{align*}
    }
  }
  \caption{\footnotesize Reduction rules on closures.}
  \label{table:reduction}
\end{table*}

\begin{definition}\label{def:rw}
The reduction rules are shown in Table~\ref{table:reduction}. The
rules split into three categories:
\subref{subtable:reduction_classical} rules handling the classical
part of the calculus; \subref{subtable:reduction_quantum} rules
dealing with quantum data; and
\subref{subtable:reduction_congruence} congruence rules for the
call-by-value strategy. Note that in the statement of the rules, $V$
and $W$ refer to values.

In the rules in
Table~\ref{table:reduction}\subref{subtable:reduction_quantum},
the quantum state $q$
has size $n$. The
quantum state $\qarray'$ in the first rule is obtained by applying the
$k$-ary unitary gate $U$ to the qubits
$\qlist(x_{1}),\dots,\qlist(x_{k})$. Precisely, $\qarray'=(\sigma\circ
(U\otimes\id)\circ\sigma^{-1})(\qarray)$, where $\sigma$ is the action
on $\C^{2^n}$ of any permutation over $\{1,\dots, n\}$ such that
$\sigma(i)=\qlist(x_i)$ whenever $i\leq k$. In the rules about
measurements, we assume that if $\qarray_0$ and $\qarray_1$ are
normalized quantum states of the form
\begin{equation}\label{eq:def-op-meas1}
{\textstyle\sum_j}\alpha_j\ket{\phi_j}\otimes\ket0\otimes\ket{\psi_j},
 ~~
{\textstyle\sum_j}\beta_j\ket{\phi_j}\otimes\ket1\otimes\ket{\psi_j},
\end{equation}
then $\qarray'_0$ and $\qarray'_1$ are respectively
\begin{equation}\label{eq:def-op-meas2}
{\textstyle\sum_j}\alpha_j\ket{\phi_j}\otimes\ket{\psi_j},
 ~~
{\textstyle\sum_j}\beta_j\ket{\phi_j}\otimes\ket{\psi_j},
\end{equation}
where the vectors $\phi_j$ have dimension $\qlist(x)-1$ (so
that the measured qubit is $\ell(x)$).  
\end{definition}

In summary, the quantum state acts as a shared global store that is
updated destructively by the various quantum operations.

Note that the only probabilistic reduction step is the one
corresponding to measurement. Also, we underline that the hypothesis
associated with a congruence rule
$\am{\qarray,\qlist,C[M]}{\redto[p]}\am{\qarray',\qlist',C[M']}$ takes
into account the whole quantum states $\qarray$ and $\qarray'$. In
fact, because of the entanglement, the evaluation of
$\am{\qarray,\qlist|_M,M}$ may have a side-effect on the state of the
qubits pointed to by the variables occurring in the context $C[\,]$.

The rules assume that the involved
closures are well-defined. In particular, whenever
$\am{\qarray,\qlist, M}\redto[p]\am{\qarray,\qlist, M'}$, the two
terms $M$ and $M'$ have the same free variables. For example, the
closure $\am{\ket{00},\ket{yz},(\lambda x.y)z}$ cannot reduce and it
represents an error: it would reduce to the erroneous quantum closure
$\am{\ket{00},\ket{yz},z}$,
where the domain of the linking function is not the set of free
variables, as specified by Definition~\ref{def:qclos}.
The type system will prevent such an error as
proven in Proposition~\ref{prop:safety}.

\begin{example}\label{ex:term_red}
  Recall Example~\ref{ex:term_type}. We have
  $\am{\ket{},\ket{},\mathtt{cointoss}}\redto[1]\am{\ket{1},\ket{x},\meas(H
    x)}\redto[1]\am{\frac1{\sqrt 2}(\ket{0}+\ket{1}),\ket{x},\meas\;
    x}$, the latter reducing to either $\am{\ket{},\ket{},\ttrue}$ or
  $\am{\ket{},\ket{},\ffalse}$, with equal probability $\frac12$. As
  for $\mathtt{entangle}$, we have that
  \begin{align*}
    &\am{\alpha\ket0+\beta\ket{1},\ket x,\mathtt{entangle}\;
      x}
    \\
    \redto[1]~~&
    \am{\alpha\ket0+\beta\ket{1},\ket
      x,N_c(x\otimes(\new\;\ffalse))}
    \\
    \redto[1]~~&
    \am{\alpha\ket{00}+\beta\ket{10},\ket{xy},N_c(x\otimes
      y)}
    \\
    \redto[1]~~&
    \am{\alpha\ket{00}+\beta\ket{11},\ket{xy},x\otimes
      y}.
  \end{align*}
  Similarly, one can check that
  $\am{\alpha\ket0+\beta\ket1,\ket q,\mathtt{qlist}\, q}$ behaves as
  described in Section~\ref{sec:rec-type}, reducing to
  $\am{\alpha\ket0+\beta\ket1,\ket q,\cons q\nil}$ with probability
  $\frac 12$, to
  $\am{\alpha\ket{00}+\beta\ket{11},\ket{qq'},\cons{q'}{\cons q
      \nil}}$ with probability $\frac 14$, etc. In particular, notice
  that in any single reduction sequence the variable $q$ has not been
  duplicated, as correctly asserted by the type of $\mathtt{qlist}$.
\end{example}

\begin{lemma}[Substitution]\label{lemma:substitution}
  Suppose $\oc\Delta,\Gamma,x:A\entail M:B$ and $\oc\Delta,\Sigma\entail
  V:A$, where $\Gamma$ and $\Sigma$ are linear contexts with disjoint
  domain. Then $\oc\Delta,\Gamma,\Sigma\entail M\{V/x\}:B$.\qed
\end{lemma}

\begin{proposition}[Subject reduction]\label{prop:subject_reduction}
  When $\am{\qarray,\ket{y_1\dots y_n},
    M}\redto[p]\am{\qarray',\ket{x_1\dots x_{n'}}, M'}$ and
  $y_1:\qubit,\dots, y_n:\qubit\entail M:A$, then $x_1:\qubit,\dots,
  x_{n'}:\qubit\entail M':A$.\qed
\end{proposition}

\begin{proposition}[Type safety]\label{prop:safety}
  If $\am{\qarray,\qlist, M}$ is typable then either $M$ is a value or
  there is a closure $\am{\qarray',\qlist', M'}$ such that
  $\am{\qarray,\qlist, M}\redto[p]\am{\qarray',\qlist', M'}$. 
  Moreover, if $M$ is not a value, the total probability of all 
  possible single-step
  reductions from $\am{\qarray,\qlist, M}$ is $1$. \qed
\end{proposition}

\begin{lemma}[Totality]\label{lemma:totality}
If $\am{\qarray,\qlist, M}\redto[p]\am{\qarray',\qlist', M'}$ 
and $\am{\qarray,\qlist, M}$ is total, then $\am{\qarray',\qlist',
  M'}$ 
is total too.
\end{lemma}

\begin{proof}
By induction on a derivation of $\am{\qarray,\qlist, M}
\redto[p]\am{\qarray',\qlist', M'}$, one proves that 
$
\dim(\qarray')=\dim(\qarray)+\dim(\qlist')-\dim(\qlist)
$
where $\dim(\qarray)$ is the size of the quantum state $\qarray$ 
and $\dim(\qlist)$ is the cardinality of the domain set of the 
linking function $\qlist$. Then, one gets the statement, since 
$\am{\qarray,\qlist, M}$ is total iff $\dim(\qarray)=\dim(\qlist)$.
\end{proof}

\begin{notation}
  The reduction relation $\redto$ defines the probability that a
  closure reduces to another one in a single step. We extend this
  relation to an arbitrary large (but finite) number of reduction
  steps with the notation
  $\Red^n_{\am{\qarray,\qlist,M},\am{q',\qlist',V}}$: it is the total
  probability of $\am{\qarray,\qlist,M}$ reducing to a value 
  $\am{q',\qlist',V}$. It is defined as the sum of all
  $\prod_{i=1}^mp_i$, where $
  \am{\qarray,\qlist,M}\redto[p_1]\am{\qarray_1,\qlist_1,M_1}
  \cdots\redto[p_m]\am{\qarray,\qlist',V}
  $ is a finite reduction sequence of $m\leq n$ steps. We write
  $\Red^\infty_{\am{\qarray,\qlist,M},\am{q',\qlist',V}}$ for the sup
  over $n$ of $\Red^n_{\am{\qarray,\qlist,M},\am{q',\qlist',V}}$.
  Finally, we define the \define{total probability}
  $\Halt_{\am{\qarray,\qlist,M}}$ of $\am{\qarray,\qlist,M}$
  converging to any value as
  $
  \sum_{\am{\qarray',\qlist',V}\in{\Val}}
  \Red_{\am{\qarray,\qlist,M},\am{\qarray',\qlist',V}}^\infty.
  $
\end{notation}

\section{Denotational semantics}\label{sec:sem}

We interpret the quantum lambda calculus in a suitable extension
$\freecat[\ccpms]$ of the category
$\CPM$ described in Section~\ref{sec:background}.
What $\CPM$ essentially misses is the linear logic exponential $\oc
A$, and our plan is to introduce it via the equation
\begin{equation}\label{eq:bang_biproduct}
\oc A\ass\bigoplus_{k=0}^\infty A^{\odot k},
\end{equation}
where $\bigoplus_{k=0}^\infty$ is the infinite biproduct of the family
$\{A^{\odot k}\}_k$, each $A^{\odot k}$ being the symmetric $k$-fold
tensor power of $A$, i.e., the equalizer of the $k!$ symmetries of the
$k$-ary tensor $A^{\otimes k}\ass A\otimes\dots\otimes A$.

The category $\CPM$ cannot express this equation
because it lacks both infinite biproducts and a convenient definition
of symmetric tensor
powers. The category $\freecat[\ccpms]$ is in some sense the minimal
extension of $\CPM$ having these two missing
ingredients. 

The plan of the section is as follows.
Section~\ref{sect:preli_cpm} presents some preliminary material.
Section~\ref{sect:thecategory} defines
$\freecat[\ccpms]$ and Section~\ref{sect:lafont_ccpms} develops the categorical
structure allowing us to interpret the quantum
lambda calculus. Section~\ref{Sect:soundness} sketches the proof of
the soundness of the model with respect to the operational semantics.
Finally, Section~\ref{sect:examples}
discusses the denotations of the programs {\bf qlist} and {\bf
  teleport}.

\subsection{Preliminaries: from \texorpdfstring{\CPM{}}{CPM} to 
  \texorpdfstring{$\ccpms$}{bar(CPMs)}}\label{sect:preli_cpm}

\paragraph{Permutation groups.} Let $\symgroup_n$ be the symmetric
group of degree $n$, i.e., the group of permutations of
$n=\{0,\dots,n-1\}$. Any permutation $g\in\symgroup_n$ gives rise to a
matrix $P_g\in\C^{n\times n}$, defined by $P_g(e_i)=e_{g(i)}$, where
$e_i$ is the $i$th standard basis vector. We define an action of $g$
on $\C^{n\times n}$ by $g\cdot M := P_gMP_g\inv$. Moreover, for a
subgroup $G\subseteq\symgroup_n$, we define
\begin{equation}\label{eq:action_permutation}
	G\cdot M := \frac{1}{\#G}\sum_{g\in G}g\cdot M,
\end{equation}
where $\#G$ is the number of elements of $G$.

\begin{lemma}
  Given a subgroup $G\subseteq\symgroup_n$, its action on $\C^{n\times
    n}$ is idempotent (i.e., $G\cdot G\cdot M=G\cdot M$ for all $M$)
  and completely positive.
\end{lemma}

\begin{proof}
  For the idempotence, notice that for every $g\in G$, $gG=G$,
  therefore: $G\cdot G\cdot M = \frac{1}{\#G}\sum_{g\in G} gG\cdot M = 
  G\cdot M$.
The complete positivity of $G$ is derived from the complete positivity
of each map $M\mapsto g\cdot M = P_gMP_g\inv$.
\end{proof}

In the sequel, we use the notation $G$ both for a subgroup of
$\symgroup_n$ and for the completely positive map defined by it. The
above Lemma allows us to define the set of completely positive maps
from $\C^{n\times n}$ to $\C^{m\times m}$ invariant under the actions
of two subgroups $G\subseteq\symgroup_n$, $H\subseteq\symgroup_m$ by
\[
\cpms(G,H)\ass\{f\in\CPM(n,m)\such G\pc f\pc H=f\},
\]
where $f;g$ is the diagrammatic composition $(f;g)(x) = g(f(x))$, and
$\CPM(n,m)$ is the set of completely positive maps from $\C^{n\times
  n}$ to $\C^{m\times m}$.

\paragraph{Completion of the L\"owner positive cone.} The set
$\cpms(G,H)$ is a module over the semi-ring $\Rp$ of the non-negative
real numbers. The L\"owner order $\sqleq$ on completely positive
maps~\cite{Selinger04}{} endows this module with the structure of a
\emph{bounded} directed complete partial order (bdcpo), i.e., there is
a minimum element (the zero function $\mathbf 0$), and any directed
set $D$ that is bounded (i.e., such that there exists $f\in
\cpms(G,H)$ such that for all $g\in D$, $g\sqleq f$) has a least upper
bound $\bigvee D\in\cpms(G,H)$. However there exist unbounded directed
subsets in $\cpms(G,H)$. We therefore need to complete $\cpms(G,H)$ to
a dcpo.

The relevant construction is the {\em D-completion} of
{\cite{ZhaoFan2010}}, which we briefly recall. Given any poset $P$,
say that a subset $X$ is {\em Scott-closed} if it is down-closed and
for every directed $I\seq S$, if the least upper bound 
$\bigvee I$ exists in $P$, then $\bigvee
I\in S$. We say that a monotone function between posets $f:P\to Q$ is
{\em Scott-continuous} if it preserves all {\em existing} least upper
bounds of directed subsets. Let $\Gamma(P)$ be the set of Scott-closed
subsets of $P$; this forms a dcpo under the subset ordering. The {\em
  D-completion} $\cs(P)$ is defined to be the smallest sub-dcpo of
$\Gamma(P)$ containing all sets of the form $\down x$. Then $\cs(P)$
is a dcpo, and there is a canonical injective Scott-continuous map
$\iota:P\to\cs(P)$, defined by $\iota(x)=\down x$, which allows us to
regard $P$ as a subset of $\cs(P)$. The D-completion preserves all
existing least upper bounds of directed sets, is idempotent, and
satisfies the following universal property: given any other dcpo $E$
and Scott-continuous map $f:P\to E$, there exists a unique
Scott-continuous $g:\cs(P)\to E$ such that $f=\iota\pc g$. It follows
that the D-completion is functorial. Moreover, if $P$ is a bounded
directed complete partial order, then $P$ is an initial subset of
$\cs(P)$, i.e., the only new elements added by the completion are ``at
infinity''. We call these the {\em infinite} elements of $\cs(P)$.

The homset $\cpms(G,H)$ is then extended by D-completion, namely,
$\overline\cpms(G,H) :=\cs(\cpms(G,H))$. The categorical operations
are extended in the unique Scott-continuous way, using the universal
property of D-completion.  This allows us to define indexed sums over
$\ccpms(G,H)$, as follows.  If $\{f_i\}_{i\in
  I}\subseteq\overline\cpms(G,H)$ is a (possibly infinite) indexed
family, $\sum_{i\in I}f_i$ is defined as $\dirsup
_{F\subseteq_{\mathrm{fin}}I}\bigl(\sum_{i\in F}f_i\bigr)$.  Indeed,
the set $\{\sum_{i\in F}f_i\;;\;F\subseteq_{\mathrm{fin}}I\}$ is
always directed, so has a least upper bound in the order completion
$\overline\cpms(G,H)$ of $\cpms(G,H)$.

\subsection{The category
  \texorpdfstring{$\freecat[\ccpms]$}{bar(CPMs)+}}
\label{sect:thecategory}
Given a set $A$ and $a,a'\in A$, define the \emph{Kronecker symbol}
$\delta_{a,a'}\in\N$ which takes value $1$ if $a=a'$ and $0$ if $a\neq
a'$.

\begin{description}
\item[Objects] are given by indexed families $\qfin A=\{(\arity{\qfin
    A}_a, \symm{\qfin A}_a)\}_{a\in\web{\qfin A}}$, where the index
  set $\web{\qfin A}$ is called the \emph{web} of $\qfin A$ and, for
  every $a\in\web{\qfin A}$, $\arity{\qfin A}_a$ is a natural
  non-negative integer, and $\symm{\qfin A}_a$ a subgroup of
  permutations of degree~$\arity{\qfin A}_a$, called respectively the
  \emph{dimension} and the \emph{permutation group} of $\qfin A_a$.
\item[Morphisms] from $\qfin A$ to $\qfin B$ are matrices $\phi$
  indexed by $\web{\qfin A}\times\web{\qfin B}$ and such that
  $\phi_{a,b}\in\ccpms(\symm{\qfin A}_a, \symm{\qfin B}_b)$.
\item[Composition] of $\phi\in\freecat[\ccpms](\qfin A,\qfin B)$ and
  $\psi\in\freecat[\ccpms](\qfin B,\qfin C)$ is the matrix
  $\phi\pc\psi$ defined by, for $a\in\web{\qfin A}$ and
  $c\in\web{\qfin C}$, $ (\phi\pc\psi)_{a,c}\ass\sum_{b\in\web{\qfin
      B}}\phi_{a,b}\pc\psi_{b,c}.  $
\item[Identity] is the diagonal matrix built with the symmetries of
  $\qfin A$, i.e., for $a,a'\in\web{\qfin A}$, $\id^\qfin
  A_{a,a'}\ass\delta_{a,a'}\symm{\qfin A}_a$.
\end{description}

The description of the objects and the morphisms as indexed families
is crucial for inferring the structure of a compact closed Lafont
category (Section~\ref{sect:lafont_ccpms}). However, it is worthwhile
to notice that $\freecat[\ccpms]$ can also be presented as a concrete
category of modules and linear maps between modules. Let us sketch
such an alternative presentation.

Let $\qfin A$ be an object of $\freecat[\ccpms]$. We define a module
$\pmatr(\qfin A)$ over $\overline{\Rp}=\Rp\cup\{\infty\}$ as
follows. For every $a$ in $\web{\qfin A}$, let us write $\pmatr(a)$
for the cone of the positive matrices in $\symm{\qfin
  A}_a(\C^{\arity{\qfin A}_a\times \arity{\qfin A}_a})$, this latter
being the subspace of the matrices in $\C^{\arity{\qfin A}_a\times
  \arity{\qfin A}_a}$ invariant under $\symm{\qfin A}_a$. This
positive cone $\pmatr(a)$ is an $\Rp$-module. We then define:
\begin{equation}\label{eq:pos_obj}
\pmatr(\qfin A)\ass
\bigoplus_{a\in\web{\qfin A}}(\cs(\pmatr(a))\}).
\end{equation}
In fact, we have that $\pmatr(a)\simeq\cpms(\symgroup_1,\symm{\qfin
  A}_a)$ and $\pmatr(\qfin A)\simeq\bigoplus_{a\in\web{\qfin
    A}}\ccpms(\symgroup_1,\symm{\qfin A}_a)$. Hence, $\pmatr(\qfin A)$
is a continuous module over $\overline{\Rp}$: addition and scalar
multiplication are defined pointwise and are continuous operations
with respect to the L\"owner order.

Let $f:\pmatr(\qfin A)\to \pmatr(\qfin B)$ be a continuous
module homomorphism.  We say that $f$ is \define{completely positive}
if all the module homomorphisms $f_{a,b} = \inj{a}\pc f\pc\proj b$ are
completely positive maps, for all $a\in\web{\qfin A}$ and
$b\in\web{\qfin B}$. (Indeed, since the positive matrices span the
complex vector space of square matrices (of corresponding size), one
can canonically extend the definition of complete positivity to module
homomorphisms $\pmatr(a)\to\pmatr(b)$).

\begin{proposition}\label{prop:concrete_cpms_pi}
  There is an isomorphism between the homset
  $\homof{\freecat[\ccpms]}{\qfin A,\qfin B}$ and the continuous
  module homomorphisms from $\pmatr(\qfin A)$ to $\pmatr(\qfin B)$
  that are completely positive. \qed
\end{proposition}

\subsection{\texorpdfstring{$\freecat[\ccpms]$}{bar(CPMs)+} 
  as a model of the quantum lambda calculus}
\label{sect:lafont_ccpms}

A compact closed category is a special case of symmetric monoidal
closed category.  A symmetric monoidal closed category with finite
products, such that each object has a corresponding free commutative
comonoid, is called a \emph{Lafont category}, which is known to be a
model of intuitionistic linear logic
\cite{lafont:these,melliespanorama}. The category $\freecat[\ccpms]$
can be endowed with such a structure, as we will show in
Sections~\ref{subsubsect:biprod}--\ref{subsubsect:exp} below.  We can
therefore interpret the quantum lambda calculus in $\freecat[\ccpms]$.

The denotation $\denot{A}$ of a type $A$ is an object of
$\freecat[\ccpms]$. In case $A$ is the ground type (i.e., $\tunit$,
$\qubit$), its denotation is:
\begin{align*}
\web{\denot{\qubit}}&\ass\{\star\},
&\arity{\denot{\qubit}}_\star&\ass 2,
&\symm{\denot{\qubit}}_\star\ass\{\id\},
\\
\web{\denot{\tunit}}&\ass\{\star\},
&\arity{\denot{\tunit}}_\star&\ass 1,
&\symm{\denot{\tunit}}_\star\ass\{\id\}.
\end{align*}
The denotation of the other types is given by structural induction,
following the compact closed Lafont structure of
$\freecat[\ccpms]$. We note in particular that the permutation groups
play a role only when interpreting $\oc$-formulas.

Let $\Gamma=x_1{:}A_1,\dots,x_n{:}A_n$. The denotation of a typing
judgement $\Gamma\entail M:A$ is a morphism $\denot{M}^{\Gamma\entail
  A}:\denot{A_1\otimes\cdots\otimes A_n}\to\denot{A}$.
The definition is by structural induction on the unique type
derivation $\pi$ of $\Gamma\entail M:A$
(see Proposition~\ref{prop:unicity_derivation}). 
\begin{table}
\[
\denot{\meas}^{\oc\Delta\vdash\qubit\multimap\bit}_{\vec m,(\ast,b)}=
(\begin{smallmatrix}
   \alpha & \beta\\
   \gamma & \delta
   \end{smallmatrix})\mapsto
   \begin{cases}
   \alpha&\text{if $\vec m=\vec{[\,]}$ and $b=\ffalse$,}\\
   \delta&\text{if $\vec m=\vec{[\,]}$ and $b=\ttrue$,}\\
   0&\text{otherwise.}
   \end{cases}
\]
\[
\denot{\new}^{\oc\Delta\vdash\bit\multimap\qubit}_{\vec m,(b,\ast)}=
\alpha\mapsto
   \begin{cases}
   (\begin{smallmatrix}
   \alpha & 0\\
   0 & 0
   \end{smallmatrix})&\text{if $\vec m=\vec{[\,]}$ and $b=\ffalse$,}\\
   (\begin{smallmatrix}
   0 & 0\\
   0 & \alpha
   \end{smallmatrix})&\text{if $\vec m=\vec{[\,]}$ and $b=\ttrue$,}\\
   \mathbf 0&\text{otherwise.}
   \end{cases}
\]
\[
\denot{U}^{\oc\Delta\vdash\qubit^{\otimes n}
  \multimap\qubit^{\otimes n}}_{\vec m,(\vec\ast,\vec\ast)}=
M\mapsto
   \begin{cases}
   U M U^{-1}&\text{if $\vec m=\vec{[\,]}$,}\\
   \mathbf 0&\text{otherwise.}
   \end{cases}
\]
\caption{\footnotesize Interpretation of the quantum constants. 
The writing $\vec m$ stands for a sequence of multisets in 
$\web{\denot{\oc\Delta}}$, the equality $\vec m=\vec{[\,]}$ 
meaning that each of these multisets is empty. $U$ and $M$ have 
the same dimension $\C^{2^n\times 2^n}$, $U$ being unitary.}
\label{table:denotation_quantum}
\end{table}
\begin{table*}
\centering
\def\myscale{0.84}
\centering
\scalebox{\myscale}{
\centering
\subfloat[$\oc\Delta,x:A\vdash x:A$]{
\centering
\parbox{100pt}{
\centering
\[
\xymatrix@R=0pt@C=1pc{
   \oc\Delta\otimes A\ar[r]^-{\weak\otimes\id}&
   \tunit\otimes A\simeq A
}
\]
\vspace{-10pt}
}
}
\subfloat[$\oc\Delta,x:\oc A\vdash x:A$]{
\centering
\parbox{100pt}{
\centering
\[
\xymatrix@R=0pt@C=1pc{
   \oc\Delta\otimes \oc A\ar[r]^-{\weak\otimes \der}&
   \tunit\otimes A\simeq A
}
\]
\vspace{-10pt}
}
}
\subfloat[$\oc\Delta,\vdash V:\oc A$]{
\centering
\parbox{120pt}{
\centering
\[
\xymatrix@R=0pt@C=1pc{
   \oc\Delta\ar[r]^{\dig}&\oc\oc\Delta\ar[r]^-{\bierman}&
   \oc(\oc\Delta)\ar[r]^-{\oc\phi}&\oc A
}
\]
\vspace{-10pt}
}
}
\subfloat[$\oc\Delta,\Gamma\vdash \lambda x^A.M:A\multimap B$]{
\centering
\parbox{150pt}{
\centering
\[
\xymatrix@R=0pt@C=2pc{
   \oc\Delta\otimes\Gamma\ar[r]^-{\cmatrix{\phi}}&A\multimap B
}
\]
\vspace{-10pt}
}
}
}

\vspace{-2ex}
\scalebox{\myscale}{
\centering
\subfloat[$\oc\Delta,\Gamma,\Sigma\vdash MN:B$]{
\centering
\parbox{240pt}{
\centering
\[
\xymatrix@R=0pt@C=1pc{
   \oc\Delta\otimes\Gamma\otimes\Sigma\ar[r]^-{\contr\otimes\id}&
   \oc\Delta\otimes\Gamma\otimes\oc\Delta\otimes
   \Sigma\ar[r]^-{\phi\otimes\psi}&
   A\otimes A\multimap B\ar[r]^-{\eval}&
   B
}
\]
\vspace{-10pt}
}
}
\subfloat[$\oc\Delta\vdash \punit:\tunit$]{
\centering
\parbox{60pt}{
\centering
\[
\xymatrix@R=0pt@C=4ex{
   \oc\Delta\ar[r]^{\weak}&\tunit
}
\]
\vspace{-10pt}
}
}
\subfloat[$\oc\Delta,\Gamma,\Sigma\vdash \letunitterm{M}{N}:A$]{
\centering
\parbox{265pt}{
\[
\xymatrix@R=0pt@C=1pc{
   \oc\Delta\otimes\Gamma\otimes\Sigma\ar[r]^-{\contr\otimes\id}&
   \oc\Delta\otimes\Gamma\otimes\oc\Delta\otimes
   \Sigma\ar[r]^-{\phi\otimes\id}&
   \tunit\otimes\oc\Delta\otimes\Sigma\simeq
   \oc\Delta\otimes\Sigma\ar[r]^-{\psi}&
   A
}
\]
\vspace{-10pt}
}
}
}

\vspace{-2ex}

\scalebox{\myscale}{
\subfloat[$\oc\Delta,\Gamma,\Sigma\vdash \tensterm MN:A\otimes B$]{
\centering
\parbox{200pt}{
\centering
\[
\xymatrix@R=0pt@C=.5pc{
   \oc\Delta\otimes\Gamma\otimes\Sigma\ar[rr]^-{\contr\otimes\id}&&
   \oc\Delta\otimes\Gamma\otimes\oc\Delta\otimes
   \Sigma\ar[rr]^-{\phi\otimes\psi}&&
   A\otimes B
}
\]
\vspace{-10pt}
}
}
\subfloat[$\oc\Delta,\Gamma,\Sigma\vdash \lettensterm{x^A}{y^B}{M}{N}:C$]{
\centering
\parbox{251pt}{
\centering
\[
\xymatrix@R=0pt@C=.5pc{
   \oc\Delta\otimes\Gamma\otimes\Sigma
   \ar[rr]^-{\contr\otimes\id}&&
   \oc\Delta\otimes\Gamma\otimes\oc\Delta
   \otimes\Sigma\ar[rr]^-{\phi\otimes\id}&&
   A\otimes B\otimes\oc\Delta\otimes\Sigma\ar[r]^-{\psi}&
   C
}
\]
\vspace{-10pt}
}
}
\subfloat[$\oc\Delta,\Gamma\vdash \injl M:A\oplus B$]{
\centering
\parbox{133pt}{
\centering
\[
\xymatrix@R=0pt@C=1pc{
   \oc\Delta\otimes\Gamma\ar[r]^-{\phi}&
   A\ar[r]^-{\inj \ell{}}&
   A\oplus B
}
\]
\vspace{-10pt}
}
}
}

\vspace{-2ex}

\scalebox{\myscale}{
\subfloat[$\oc\Delta,\Gamma\vdash \injr M:A\oplus B$]{
\centering
\parbox{140pt}{
\centering
\[
\xymatrix@R=0pt@C=1pc{
   \oc\Delta\otimes\Gamma\ar[r]^-{\phi}&
   B\ar[r]^-{\inj r{}}&
   A\oplus B
}
\]
\vspace{-10pt}
}
}
\subfloat[$\oc\Delta,\Gamma,\Sigma\vdash \match{M}{x^A}{N}{y^B}{L}:C$
          ]{
\centering
\parbox{449pt}{
\centering
\[
\xymatrix@R=0pt@C=.5pc{
   \oc\Delta\otimes\Gamma\otimes\Sigma\ar[rr]^-{\contr\otimes\id}&&
   \oc\Delta\otimes\Gamma\otimes\oc\Delta
   \otimes\Sigma\ar[rr]^-{\psi\otimes\id}&&
   (A\oplus B)\otimes\oc\Delta\otimes\Sigma\ar[rr]^-{\pdistr}&&
   (A\otimes\oc\Delta\otimes\Sigma)\oplus(B\otimes\oc\Delta\otimes\Sigma)
   \ar[rr]^-{\phi_A\oplus\phi_B}&&
   C
}
\]
\vspace{-10pt}
}
}
}

\vspace{-2ex}

\scalebox{\myscale}{
\subfloat[$\oc\Delta,\Gamma\vdash M:\tlist A$]{
\centering
\parbox{240pt}{
\centering
\[\xymatrixcolsep{1pc}
\xymatrix{
   \oc\Delta\otimes\Gamma\ar[r]^-{\phi}&
   \tunit\oplus(A\otimes\tlist A)\ar[rr]^-{\id\oplus\pdistr}&&
   \tunit\oplus(\bigoplus_{n=1}^\infty A^{\otimes n})=\tlist A
}
\]
\vspace{-15pt}
}
}
\subfloat[$\bang\Delta,\Gamma\entail\letrec{f}{x}{M}{N}:C$]{
\centering
\parbox{350pt}{
\centering
\[
\oc\Delta\tensor\Gamma
\xrightarrow{\contr}
\oc\Delta\tensor\Gamma\tensor\oc\Delta
\xrightarrow{\id\tensor\fixpoint(\dig;\bierman;\oc{(\Lambda\phi)})}
\oc\Delta\tensor\Gamma\tensor\oc{(A\loli B)}
\xrightarrow{\psi}
C
\]
\vspace{-10pt}
}}
}
\caption{\footnotesize Sketch of the interpretation of the typing 
  judgements, using the Lafont structure of $\freecat[\ccpms]$
  defined in Section~\ref{sect:lafont_ccpms}. The
  morphisms $\phi$, $\psi$, $\phi_A$, $\phi_B$ refer to the denotation
  of the premises of the unique derivation concluding a typing
  judgement. In (c) and (n), the morphism $\bierman$ stands for 
  $\bierman^\unit$ or the suitable sequence of $\bierman^\otimes$,
  depending on the context $\oc\oc\Delta$.}\label{table:denotation_rule}
\end{table*}
The denotations of the constants $\meas$, $\new$ and the unitary
transformations are given in
Table~\ref{table:denotation_quantum}. Table~\ref{table:denotation_rule}
briefly recalls the denotation of the usual linear logic rules. Here,
the morphisms $\phi, \psi, \phi_A, \phi_B$ refer to the denotation of
the premises of the last rule of $\pi$, which are uniquely defined
given $\Gamma\entail M:A$.

In the interpretation of the $\mathtt{letrec}$ constructor, the fixed
point operator $\fixpoint$ is defined as follows. Let $\phi$ be a
morphism in the set $\homof{\freecat[\ccpms]}{\oc{C}\tensor\oc{A},\oc A}$. By
induction on $n$, we define the morphism
$\phi^n\in\homof{\freecat[\ccpms]}{\oc{C},\oc A}$: $
\phi^0\ass\oc{C}\xrightarrow{\weak;\oc{\bf 0}}\bang{A}$,
$\phi^{n+1}\ass\oc{C}\xrightarrow{\contr}\bang{C}\tensor\bang{C}
\xrightarrow{\id\tensor
  \phi^n} \bang{C}\tensor\bang{A}\xrightarrow{\phi}\bang{A}$.  Since
$\phi$ can be regarded as a continuous module homomorphism (in
particular it is monotone), the set $\{\phi^n\}$ is directed
complete. We define $\fixpoint(\phi)$ as its least upper bound.

\subsubsection{Biproduct (\texorpdfstring{$\qfin A\oplus\qfin B$}{A + B})}
\label{subsubsect:biprod}  

Let $I$ be a (possibly infinite) set of indexes. The biproduct
$\bigoplus_{i\in I}\qfin A_i$ of a family $\{\qfin A_i\}_{i\in I}$ of
objects in $\freecat[\ccpms]$ is defined by
\[
\web{\bigoplus_{i\in I}\qfin A_i}
\ass\!\bigcup_{i\in I}\{i\}\times\web{\qfin A_i},\;\;
\arity{\bigoplus_{i\in I}\qfin A_i}_{(j,a)}\!\!\ass\arity{\qfin A_j}_a,
\;\;
\symm{\bigoplus_{i\in I}\qfin A_i}_{(j,a)}\!\!\ass\symm{\qfin A_j}_a.
\]
The corresponding projections and injections are denoted respectively by
$\proj j$ and $\inj j$ and defined as:
\begin{align*}
\proj j_{(i,a),a'}&\ass\inj j_{a',(i,a)}\ass
\delta_{j,i}\delta_{a,a'}\symm{\qfin A_i}_a.
\end{align*}
The tupling $\fprod{\phi_i}{i\in I}$ (resp.\ (co)-tupling
$\fcoprod{\psi_i}{i\in I}$) of a family of morphisms
$\phi_i$ elements of $\homof{\freecat[\ccpms]}{\qfin A,\qfin B_i}$ (resp.\
$\psi_i$ elements of $\homof{\freecat[\ccpms]}{\qfin A_i,\qfin B}$)
is defined by
$(\fprod{\phi_i}{i\in
    I})_{a,(j,b)}\ass(\phi_{j})_{a,b}$ (resp. $(\fcoprod{\psi_i}{i\in
    I})_{(j,a),b}\ass(\psi_{j})_{a,b}$).

\begin{example}\label{ex:bool_sem}
  Recall that in Notation~\ref{notation:syntax}, the type $\bit$ is
  interpreted as the biproduct $\denot{\tunit}\oplus\denot{\tunit}$,
  which is the two-element family
  $\{(1,\{\id\})_{\ttrue},(1,\{\id\})_{\ffalse}\}$.  The positive
  cones associated with $\tunit$ and $\bit$ are:
  $\pmatr(\denot{\tunit})=\overline{\Rp}$ and
  $\pmatr(\denot{\bit})=\overline{\Rp}^2$.

The typing judgement $\vdash \ttrue:\bit$ is interpreted as the right
injection, which can be seen both as a family of two completely
positive maps from $\C$ to $\C$ (i.e., $\denot{\ttrue}^{\vdash
  \bit}_{\star,\ttrue}=p\mapsto p$ and $\denot{\ttrue}^{\vdash
  \bit}_{\star,\ffalse}=p\mapsto 0$) and as a quantum compatible and
completely positive map sending $p\in\overline{\Rp}$ to
$(0,p)\in\overline{\Rp}^2$. Symmetrically,
$\denot{\ffalse}^{\vdash\bit}$ is the map $p\mapsto (p,0)$.

As an example of a term with free variables, consider
$\texttt{Neg}_x\ass\iftermx{x}{\ffalse}{\ttrue}$. The denotation of
$x:\bit\vdash \texttt{Neg}_x:\bit$ can be seen both as a family of
four constant maps $\denot{\texttt{Neg}_x}^{\bit\vdash\bit}_{b,b'}$
from $\C$ to $\C$ of value $1$ if $b\neq b'$ and $0$ otherwise, and as
a single map from $\overline\Rp^2$ to $\overline\Rp^2$ sending
$(p,p')$ to $(p',p)$.
\end{example}

\subsubsection{Symmetric monoidal structure (\texorpdfstring{$\qfin
    A\otimes \qfin B$}{A tensor B},
  \texorpdfstring{$\unit$}{1} and 
  \texorpdfstring{$\tlist{\qfin A}$}{list(A)})}
\label{subsubsect:smc}

The bifunctor
$\tensor:\freecat[\ccpms]\times\freecat[\ccpms]\to\freecat[\ccpms]$ is
defined on objects $\qfin A, \qfin B$ by:
\[
\web{\qfin A\tensor\qfin B}\ass\web{\qfin A}\times\web{\qfin B},
\qquad
\arity{\qfin A\tensor\qfin B}_{(a,b)}\ass\arity{\qfin A}_a
\times\arity{\qfin B}_b,\]
\[
\symm{\qfin A\tensor\qfin B}_{(a,b)}\ass\{(g,h)\;;\;g\in
\symm{\qfin A}_a, h\in\symm{\qfin B}_b\},
\]
\noindent where $\arity{\qfin A}_a\times\arity{\qfin B}_b$ is the
multiplication of the two numbers $\arity{\qfin A}_a$ and
$\arity{\qfin B}_b$, which can be seen as the lexicographically
ordered set of pairs $(i,j)$, for $i<\arity{\qfin A}_a$,
$j<\arity{\qfin B}_b$. Hence, the action of a permutation
$(g,h)\in\symm{\qfin A\tensor\qfin B}_{(a,b)}$ on $\arity{\qfin
  A}_a\times\arity{\qfin B}_b$ can be described as
$(i,j)\mapsto(g(i),h(j))$.

The bifunctor $\tensor$ on morphisms is defined componentwise, using
the standard tensor of the category $\CPM$ extended to the infinite
elements by the universal property of the D-completion
(Section~\ref{sect:preli_cpm}).
The tensor unit is the object $\denot{\tunit}$ interpreting the unit type.

The associativity, unit, and symmetry isomorphisms are defined
componentwise from the corresponding isomorphisms in $\CPM$, composed
with the actions of the groups of the objects. E.g., the symmetry is $
\sigma^{\qfin A,\qfin
  B}_{(a,b),(b',a')}\!\ass\delta_{a,a'}\delta_{b,b'}\symm{\qfin
  A\otimes\qfin B}_{(a,b)}\pc\swap^{\arity{\qfin A}_a,\arity{\qfin
    B}_b}
$, 
where $\swap^{\arity{\qfin A}_a,\arity{\qfin B}_b}$ is the symmetry in
\CPM{} between $\C^{\arity{\qfin A}_a\times\arity{\qfin
    A}_a}\otimes\C^{\arity{\qfin B}_b\times\arity{\qfin B}_b}$ and
$\C^{\arity{\qfin B}_b\times\arity{\qfin B}_b}\otimes\C^{\arity{\qfin
    A}_a\times\arity{\qfin A}_a}$.  Notice that it is sufficient to
pre-compose $\swap^{\arity{\qfin A}_a,\arity{\qfin B}_b}$ with
$\symm{\qfin A\otimes\qfin B}_{(a,b)}$ (or, symmetrically,
post-compose with $\symm{\qfin B\otimes\qfin A}_{(b,a)}$), in order to
have a map invariant under both the permutation groups $\symm{\qfin
  A\otimes\qfin B}_{(a,b)}$ and $\symm{\qfin B\otimes\qfin
  A}_{(b,a)}$. This is because $\symm{\qfin A\otimes\qfin
  B}_{(a,b)}\pc\swap^{\arity{\qfin A}_a,\arity{\qfin B}_b}=\symm{\qfin
  A}_{a}\otimes\symm{\qfin B}_{b}\pc\swap^{\arity{\qfin
    A}_a,\arity{\qfin B}_b}=\swap^{\arity{\qfin A}_a,\arity{\qfin
    B}_b}\pc \symm{\qfin B}_{b}\otimes\symm{\qfin
  A}_{a}=\swap^{\arity{\qfin A}_a,\arity{\qfin B}_b}\pc\symm{\qfin
  B\otimes\qfin A}_{(b,a)}$. Similar simplifications will be done
henceforth without explicitly mentioning it.

\begin{example}\label{ex:ifthen_sem}
  The denotation of $\qubit\otimes\qubit$ is the singleton web family
  $\{(4,\{\id\})_\star\}$. This object is associated with the cone of
  positive matrices of dimension $4\times 4$ plus the infinite
  elements needed to complete the L\"owner order.  The
  denotation of $\bit\otimes\bit$ instead has a web of cardinality
  $4$, i.e.,
  $\{(\ffalse,\ffalse),(\ffalse,\ttrue),(\ttrue,\ffalse),(\ffalse,\ffalse)\}$,
  and, for each index $b\in\web{\denot{\bit\otimes\bit}}$, we have
  $\arity{\denot{\bit\otimes\bit}}_b=1$ and
  $\symm{\denot{\bit\otimes\bit}}_b=\{\id\}$. This object is
  associated with the biproduct
  $\overline\Rp\oplus\overline\Rp\oplus\overline\Rp\oplus\overline\Rp$.
\end{example}

Notice that in the above example the tensor product distributes over
the biproducts:
$\denot{\bit\otimes\bit}=\denot{(\tunit\oplus\tunit)
\otimes(\tunit\oplus\tunit)}=\denot{\tunit\oplus\tunit\oplus
\tunit\oplus\tunit}$. This
is true in general: the isomorphism between $\qfin
A\otimes(\bigoplus_{i\in I}\qfin B_i)$ and $\bigoplus_{i\in I}(\qfin
A\otimes\qfin B_i)$ is
\[
\pdistr_{(a,(i,b)),(i',(a',b'))}\ass\delta_{i,i'}\delta_{a,a'}
\delta_{b,b'}\symm{\qfin A\otimes\qfin B_i}_{(a,b)}.
\]
This isomorphism allows us to define the list constructor as the
infinite biproduct of tensor powers
$
\tlist{\qfin A}\ass\bigoplus_{n=0}^\infty\qfin A^{\otimes n}.
$
In fact, we have $\tlist{\qfin A}\simeq\unit\oplus(\qfin A\otimes
\tlist{\qfin A})$.

\begin{example}\label{ex:list}
  The denotation of the unit type list is:
  $\web{\denot{\tlist\unit}}=\N$ and, for every $n\in\N$,
  $\arity{\denot{\tlist\unit}}_n=1$,
  $\symm{\denot{\tlist\unit}}_n=\{\id\}$. This object can be
  associated with the module $\overline{\Rp}^\N$ and is suitable for
  denoting the numerals in unary notation. Indeed, writing $\underline
  n$ for the list $\cons{\punit}{\dots\cons{\punit}{\nil}}$ of length
  $n$, we have $\denot{\underline
    n}^{\vdash\tlist\unit}=p\mapsto(\underbrace{0,\dots,0}_{n-1\text{
      times}},p,0,\dots)$.
\end{example}

\subsubsection{Compact closure (\texorpdfstring{$\qfin A^\perp, \qfin
    A\multimap B$}{dual(A) and A -o B})}
\label{subsubsect:cc}  

Dual objects coincide: we have 
$
	\dual{\qfin A}\ass\qfin A.
$
The unit $\eta^{\qfin A}\in\freecat[\ccpms](\unit,\dual{\qfin
  A}\otimes \qfin A)$ and co-unit $\epsilon^{\qfin
  A}\in\freecat[\ccpms](\qfin A\otimes\dual{\qfin A},\unit)$ are
defined componentwise composing the unit and co-unit of $\CPM$ with
the correspondent permutation group. Writing $E_{i,j}$ for the matrix
that has $0$ everywhere except $1$ at $(i,j)$, we have:
\begin{align*}
  \eta^{\qfin A}_{\star,(a,a')}&
  \ass 1\mapsto \sum_{i,j<\arity{\qfin A}}\symm{\qfin
    A}_a(E_{i,j})\otimes
  \symm{\qfin A}_a(E_{i,j})\\
  \epsilon^{\qfin A}_{(a,a'),\star}
  &\ass(E_{i,j}\otimes E_{i',j'})
  \mapsto\!\!\!\!\sum_{g,g'\in \symm{\qfin A}_a}\frac{1}{
    \#\symm{\qfin A}_a}\delta_{g(i),g'(i')}\delta_{g(j),g'(j')}.
\end{align*}

Compact closed categories are monoidal closed. Let us recall the
monoidal closure structure, which is needed to model the abstraction
and the application of the quantum lambda calculus.
The internal hom object is defined as 
$
\qfin A\multimap\qfin B\ass(\dual{\qfin A}\otimes \qfin B)=
\qfin A\otimes \qfin B.
$
The evaluation morphism $\Eval^{\qfin A,\qfin
  B}:\homof{\freecat[\ccpms]}{(\qfin A\multimap\qfin B)\otimes\qfin A,
  \qfin B}$ and the currying isomorphism $\Lambda(\--)$ from
$\homof{\freecat[\ccpms]}{\qfin C\otimes\qfin A,\qfin B}$ to
$\homof{\freecat[\ccpms]}{\qfin C,\qfin A\multimap\qfin B}$ are,
\[
\Eval^{\qfin A,\qfin
  B}\!\!\!\ass\sigma\pc\alpha\pc(\epsilon\tensor\id)\pc
\lambda,~~\cmatrix{\phi}\ass\lambda^{\!\!-1}\!\pc(\eta\otimes\id)\pc
\alpha^{\!\!-1}\!\pc(\id\otimes(\sigma;\phi)),
\]
\noindent where $\alpha$, $\lambda$, and $\sigma$ are the associative,
left unit and symmetric isomorphisms associated with $\otimes$.

\begin{example}\label{ex:lambda_sem}
  Let us consider the abstraction $\lambda x.\texttt{Neg}_x$ of the
  term $\texttt{Neg}_x$ discussed in Example~\ref{ex:bool_sem}. The
  denotation $\denot{\lambda x.\texttt{Neg}_x}^{\vdash
    \bit\multimap\bit}$ is obtained from
  $\denot{\texttt{Neg}_x}^{x:\bit\vdash\bit}$ just by shifting the
  matrix indexes: $\denot{\lambda x.\texttt{Neg}_x}^{\vdash
    \bit\multimap\bit}_{\star,(b,b')}=
  \denot{\texttt{Neg}_x}^{x:\bit\vdash\bit}_{b,b'}$.
  Looking at this matrix as a module homomorphism, the map
  $\denot{\lambda x.\texttt{Neg}_x}^{\vdash \bit\multimap\bit}$ is
  $p\mapsto (0,p,p,0)$, which is a map from $\overline\Rp$ to
  $\overline\Rp_{(\ffalse,\ffalse)}\oplus
  \overline\Rp_{(\ffalse,\ttrue)}\oplus
  \overline\Rp_{(\ttrue,\ffalse)}\oplus\overline\Rp_{(\ttrue,\ttrue)}$,
  where we make explicit the correspondence between the web elements
  of $\denot{\bit\multimap\bit}$ and the components of the biproduct
  associated with.

Application corresponds basically to matrix multiplication. For example,
$\denot{(\lambda x.\texttt{Neg}_x)(\meas\,
  y)}^{y:\qubit\vdash\bit}_{\star,b}$ is the function defined as
$
\sum_{b'\in\{\ttrue,\ffalse\}}\denot{\lambda
  x.\texttt{Neg}_x}_{\star,(b',b)}^{\vdash\bit\multimap\bit}\denot{\meas\,
  y}_{\star,b'}^{y:\qubit\vdash\bit}$,
which is sending $
(\begin{smallmatrix}
   \alpha & \beta\\
   \gamma & \delta
   \end{smallmatrix})$
to $\delta$ if $b=\ffalse$, $\alpha$ if $b=\ttrue$, and $0$ otherwise.
\end{example}

\subsubsection{Free commutative comonoids (\texorpdfstring{$\qfin
    A^{\odot k}$}{symmetric powers}, \texorpdfstring{$\oc \qfin A$}{!A})}
\label{subsubsect:exp} 

Let us now focus on the crucial structure modeling the linear logic
modality $\oc$. We first define the notion of $k$-th symmetric power
of an object and then we show how the biproduct of all such symmetric
powers yields an exponential structure.

\begin{notation}\label{notation:multisets}
  Given a set $X$, a \emph{multiset} $\mu$ over $X$ is a function
  $X\mapsto\N$. The \emph{support} of $\mu$ is the set
  $\supp\mu=\{a\mid\mu(a)\neq 0\}\subseteq X$, the \emph{disjoint
    union} is $(\mu\uplus\nu)(a)=\mu(a)+\nu(a)$, and the \emph{empty
    multiset} is the zero constant function. The \emph{cardinality} of
  $\mu$ is $\sum_{a\in X}\mu(a)\in\N\cup\{\infty\}$. A multiset is
  finite if it has finite cardinality. $\M_k(X)$ (resp.\ $\M_f(X)$) is
  the set of the multisets over $X$ with cardinality $k$ (resp.\
  finite). Finite multisets can be denoted by listing the occurrences
  of their elements between square brackets, i.e., $\mu=[a,a,b]$ is
  $\mu(a)=2$, $\mu(b)=1$ and zero on the other elements, and $[\,]$ is
  the empty multiset.
\end{notation}

In a symmetric monoidal category, given a natural number $k$, the
\emph{$k$-th symmetric power} of an object $\qfin A$ is a pair $(\qfin
A^{\odot k},\eq{\qfin A^{\odot k}})$ of an object $\qfin A^{\odot k}$
and a morphism $\eq{\qfin A^{\odot k}}$ from $\qfin A^{\odot k}$ to
$\qfin A^{\otimes k}$, which is an equalizer of the $k!$ symmetries of
the $k$-ary tensor $\qfin A^{\otimes k}$. Such equalizers do not exist
in general, but they do exist in $\freecat[\ccpms]$ and can be
concretely represented using the multisets notation, as follows:
\begin{gather*}
  \web{\qfin A^{\odot k}}\ass\M_k(\web{\qfin A}),\quad \arity{\qfin
    A^{\odot k}}_\mu\ass\prod_{a\in\supp\mu}(\arity{\qfin
    A}_a)^{\mu(a)},
  \\
  \symm{\qfin A^{\odot
      k}}_\mu\ass\{(h_a,g_a^1,\dots,g_a^{\mu(a)})_{a\in\supp\mu}\;;\;
  h_a\in\symgroup_{\mu(a)}, g_a^i\in\symm{\qfin A}_a\},
\end{gather*}
where $(h_a,g_a^1,\dots,g_a^{\mu(a)})_{a\in\supp\mu}$ is a
$\supp\mu$-indexed family of sequences of permutations and
$\symm{\qfin A^{\odot k}}_\mu$ is a group (composition being defined
componentwise) whose action on $\C^{\arity{\qfin A^{\odot
      k}}_\mu\times\arity{\qfin A^{\odot k}}_\mu}$ can be described by
seeing $\arity{\qfin A^{\odot k}}_\mu$ as the set of families of
sequences of the form $(i_a^1,\dots,i_a^{\mu(a)})_{a\in\supp\mu}$,
with $i_a^j<\arity{\qfin A}_a$ for every $j\leq\mu(a)$. Then, the
action of $(h_a,g_a^1,\dots,g_a^{\mu(a)})_{a\in\supp\mu}$ on such
families is:
\[
(i_a^1,\dots,i_a^{\mu(a)})_{a\in\supp\mu}\mapsto
(g_a^1(i_a^{h_a(1)}),\dots,g_a^{\mu(a)}(i_a^{h_a(\mu(a))}))_{a\in\supp\mu}.
\]
The morphism $\eq{\qfin A^{\odot k}}$ is given by
\[
\eq{\qfin A^{\odot k}}_{\mu,(a_1,\dots, a_k)}\ass
	\begin{cases}
		\symm{\qfin A^{\odot k}}_\mu&\text{if $\mu=[a_1,\dots,a_k]$,}\\
		\mathbf 0&\text{otherwise.}
	\end{cases}
\]

\begin{remark}
  The object $\denot{A}^{\odot k}$ describes $k$ {\em unordered} uses
  of an element of type $A$. The fact that our model uses the
  symmetric tensor power $\qfin A^{\odot k}$ instead of the $k$-fold tensor
  $\qfin A^{\otimes k}$ means operationally that the behavior of a program
  calling its input $k$ times does not depend on the order of the calls.
\end{remark}

\begin{example}\label{ex:symmetric_two_power}
  In Example~\ref{ex:ifthen_sem}, we have seen that
  $\denot{\qubit}^{\otimes 2}=\{(4,\{\id\})_\star\}$. The symmetric
  $2$-power $\denot{\qubit}^{\odot 2}$ is instead the singleton web
  family $\{(4,\{\id,\swap\})_\star\}$, where $4$ is represented as
  the lexicographically ordered set $\{(0,0),(0,1),(1,0),(1,1)\}$ and
  the permutation $\swap$ acts on it by $(b,b')\mapsto(b',b)$. The
  group of permutations $\{\id,\swap\}$ shrinks the set of possible
  morphisms to or from $\denot{\qubit}^{\odot 2}$. For example, the
  matrix $N_c$ associated with the controlled-not gate
  (Equation~\eqref{eq:h_and_nc}) defines a complete positive endo-map
  of $\C^{4\times 4}$, which is an endo-morphism of
  $\denot{\qubit}^{\otimes 2}$ but not of $\denot{\qubit}^{\odot 2}$,
  because $N_c$ is not invariant under the action of $\{\id,\swap\}$:
\[
\{\id,\swap\}(N_c)=\frac12(\id(N_c)+\swap(N_c))=
\frac12
\left(\begin{smallmatrix}
2&0&0&0\\
0&1&0&1\\
0&0&1&1\\
0&1&1&0
\end{smallmatrix}
\right)\neq N_c.
\]
Concerning the module associated with symmetric tensor powers,  
$\pmatr(\denot{\qubit}^{\odot 2})$ is the D-completion of
\[\left\{
	\left(
	\begin{smallmatrix}
	\alpha_1&\alpha_2&\alpha_2&\alpha_3\\
	\alpha_4&\alpha_5&\alpha_6&\alpha_7\\
	\alpha_4&\alpha_6&\alpha_5&\alpha_7\\
	\alpha_8&\alpha_9&\alpha_9&\alpha_{10}
	\end{smallmatrix}
	\right)\text{ positive}\;;\;\forall i, \alpha_i\in\C
\right\}
\]
which is a subcone of the positive cone of $\C^{4\times 4}$ of
dimension $10$.

Concerning biproducts, the denotation of $\qubit\oplus\qubit$ is given
by $\{(2,\{\id\})_{\ttrue}, (2,\{\id\})_{\ffalse}\}$, while its
symmetric tensor power $\denot{\qubit\oplus\qubit}^{\odot 2}$ is given
by the three-element family
$\{(4,\{\id,\sigma\})_{[\ttrue,\ttrue]},(4,\{\id\})_{[\ttrue,
  \ffalse]},(4,\{\id,\sigma\})_{[\ffalse,\ffalse]}\}$. Notice the
difference between the pair $(4,\{\id\})$ associated with $[\ttrue,
\ffalse]$ and the pair $(4,\{\id,\sigma\})$ associated with the two
multisets of singleton support.
\end{example}

The biproduct $\oc\qfin A\ass\bigoplus_{k=0}^\infty\qfin A^{\odot k}$
of all symmetric powers of $\qfin A$ can be defined as
\begin{align*}
\web{\oc\qfin A}&=\M_f(\web{\qfin A}),&
\arity{\oc\qfin A}_\mu&=\arity{\qfin A^{\odot k}}_\mu\!\!,&
\symm{\oc\qfin A}_\mu&=\symm{\qfin A^{\odot k}}_\mu&
\!\!\!\!\text{($\mu\in\M_k(\web{\qfin A})$)}
\end{align*}
This object yields a concrete representation of the free commutative
comonoid generated by $\qfin A$. The counit (also called
\emph{weakening}) $\weak\in\homof{\freecat[\ccpms]}{\oc\qfin A,\unit}$
and the comultiplication (or \emph{contraction})
$\contr\in\homof{\freecat[\ccpms]}{\oc\qfin A,\oc\qfin
  A\otimes\oc\qfin A}$ are:
\begin{align*}
\weak_{\mu,\star}&\ass\delta_{\mu,[\,]}\symm{\oc\qfin A}_{[]},&
\contr_{\mu,(\mu',\mu'')}&\ass\delta_{\mu,\mu'+\mu''}\symm{\oc\qfin A}_\mu.
\end{align*}

The freeness of the comonoid gives the structure of exponential
comonad. The functorial promotion maps an object $\qfin A$ to
$\oc\qfin A$ and a morphism $\phi\in\homof{\freecat[\ccpms]}{\qfin
  A,\qfin B}$ to $\oc\phi\in\homof{\freecat[\ccpms]}{\oc\qfin
  A,\oc\qfin B}$ defined by, for $\mu\in\M_f(\web{\qfin A})$ and
$\nu=[b_1,\dots,b_k]\in\M_f(\web{\qfin B})$,
\[
\oc\phi_{\mu,\nu}\ass\sum_{
  \substack{
    (a_1,\dots,a_k),\text{ st}\\
    [a_1,\dots,a_k]=\mu
  }
}
\symm{\oc\qfin{A}}_{\mu}\pc\bigotimes_{i=1}^k\phi_{a_i,b_i}
\pc\symm{\oc\qfin
  B}_\nu.
\]
The counit of the comonad (or \emph{dereliction})
$\der\in\homof{\freecat[\ccpms]}{\oc\qfin A,\qfin A}$ and the
comultiplication (or \emph{digging})
$\dig\in\homof{\freecat[\ccpms]}{\oc\qfin A,\oc\oc\qfin A}$ are
\begin{align*}
\der_{\mu,a}&\ass\delta_{\mu,[a]}\symm{\qfin A}_a,&
\dig_{\mu,M}&\ass\delta_{\mu,\sum M}\symm{\qfin \oc\qfin A}_\mu,
\end{align*}
where $M\in\web{\oc\oc\qfin A}$ is a multiset of multisets $\nu$ over
$\web{\qfin A}$ and $\sum M\in\web{\oc\qfin A}$ is the multiset union
of such $\nu$'s, i.e., for every $a\in\web{\qfin A}$, $\sum
M(a)=\sum_{\nu\in\supp M}\nu(a)^{M(\nu)}$.

Finally, the last two morphisms that are essential to interpret our
calculus are Bierman's
$\bierman^\otimes\in\homof{\freecat[\ccpms]}{\oc\qfin A\otimes\oc\qfin
  B,\oc(\qfin A\otimes\qfin B)}$ and
$\bierman^\unit\in\homof{\freecat[\ccpms]}{\unit,\oc\unit}$, given by
$
\bierman^\otimes_{(\mu,\nu),\eta}\ass\delta_{\eta,\mu\times\nu}\symm{\oc(\qfin
  A\otimes\qfin B)}_\eta$ and $
\bierman^\unit_{\star,\mu}\ass\delta_{\mu,[\star]}\symm{\unit}_\mu$,
where $\mu\times\nu$ is the multiset in $\web{\oc{(\qfin A\otimes\qfin
    B)}}$ defined by, $\mu\times\nu(a,b)\ass\mu(a)\nu(b)$.

\begin{example}\label{ex:bang_sem}
  Using the isomorphism between $\M_f(\{\star\})$ and the set $\N$, and
  between $\M_f(\{\ttrue, \ffalse\})$ and $\N\times\N$, the free
  commutative comonoids associated with $\denot{\tunit}$ and
  $\denot{\bit}$ are $\oc\denot{\tunit}=\{(1,\{\id\})_{n}\}_{n\in\N}$,
  and $\oc\denot{\bit}=\{(1,\{\id\})_{(n,m)}\}_{n,m\in\N}$.
  In general, notice that all constructions of the Lafont category
  preserve the underlying pair $(1,\{\id\})$ and act only at the level
  of webs. For more involved examples, one should look for objects with
  larger dimension, like $\denot{\qubit}$. For example,
  $
  \oc\denot{\qubit}=\{(2^n,\symgroup_n)_n\}_{n\in\N}.
  $
  Notice that $\oc\tunit$, $\oc\bit$ and $\oc\qubit$ are not allowed
  by our type grammar. In fact, $\oc\qubit$ is meaningless because of
  the no-cloning constraint on quantum bits. However, such spaces
  should exist in the model since they are isomorphic to the
  denotations of legal types, like $\oc(\tunit\multimap\tunit)$,
  $\oc(\tunit\multimap\bit)$ and $\oc(\tunit\multimap\qubit)$.
\end{example}

\subsection{The soundness theorem}\label{Sect:soundness}

The soundness of $\freecat[\ccpms]$ with respect to the operational
semantics given in Figure~\ref{table:reduction} is an easy consequence
of the fact that the category gives a (dcpo-enriched) model of linear
logic. In fact, the operational semantics is a trivial extension of a
head-reduction strategy of linear logic cut-elimination.

\begin{proposition}\label{prop:lafont}
  The category $\freecat[\ccpms]$ is a dcpo-enriched compact closed
  Lafont category, hence $\freecat[\ccpms]$ is a model of linear
  logic.
\end{proposition}

\begin{proof}[Proof (Sketch)]
  This basically amounts to showing that $\freecat[\ccpms]$ is the
  result of a categorical construction applied to $\CPMs$ which is
  known to give, under certain circumstances, a dcpo-enriched Lafont
  category and to preserve the compact closed structure of
  $\CPMs$. This construction was sketched in
  \cite{Girard99coherentbanach} and detailed
  in~\cite{MelliesTT09,LairdMM12,LairdMMP13}. It consists in moving:
  (i) from $\CPMs$ to a category $\cpms$ with symmetric tensors, which
  is actually a full sub-category of the Karoubi envelope of $\CPMs$;
  (ii) to a dcpo-enriched category $\ccpms$ using the D-completion
  defined in {\cite{ZhaoFan2010,Keimel2009}}; and, finally, (iii)
  constructing the free biproduct completion $\freecat[\ccpms]$ of
  $\ccpms$ and applying Equation~\eqref{eq:bang_biproduct}.
\end{proof}

Given a linking $\ell=\ket{y_1,\dots,y_m}$, we write $\ell\vdash M:A$
for the judgement $y_1:\qubit,\dots,y_m:\qubit\vdash M:A$.

\begin{proposition}[Invariance of the interpretation]
  \label{prop:semantic_invariance}
  Let $\ell$ be the linking $\ket{y_1,\dots,y_m}$, and assume
  $\ell\vdash M:A$. If $M$ is not a value, then for all quantum states
  $\qarray\in{\C^{2^m}}$,
\begin{equation}\label{eq:soundness}
\denot{M}^{\ell\vdash A}(\qarray\qarray^\ast)=\sum_{
   \begin{subarray}{c}
   \am{\qarray,\ell,M}\redto[p]\am{\qarray',\ell',N}
   \end{subarray}
   }p\cdot\denot{N}^{\ell'\vdash A}(\qarray'{\qarray'}^\ast).
\end{equation}
\end{proposition}

\begin{proof}
  By hypothesis, $\am{\qarray,\qlist,M}$ is a typable total closure,
  and so, by Proposition~\ref{prop:subject_reduction} and
  Lemma~\ref{lemma:totality}, all of its reducts
  $\am{\qarray',\qlist',N}$ are typable total closures, so that
  $\denot{N}^{\ell'\vdash A}(\qarray'{\qarray'}^\ast)$ is
  well-defined.

  Equation~\ref{eq:soundness} is proven by cases, depending on the
  rule applied to $\am{\qarray,\ell,M}$. The cases of
  Table~\ref{table:reduction}\subref{subtable:reduction_classical}
  follows from the fact that $\freecat[\ccpms]$ is a dcpo-enriched
  model of linear logic. The quantum rules
  (Table~\ref{table:reduction}\subref{subtable:reduction_quantum}) are
  trivial consequences of Table~\ref{table:denotation_quantum}, and
  the congruence rules of
  Table~\ref{table:reduction}\subref{subtable:reduction_congruence}
  are done by induction on $M$, using the fact that the category
  $\freecat[\ccpms]$ is linear.
\end{proof}

\begin{corollary}\label{cor:invariance_interpretation}
  We have
  $\denot{M}^{\vdash\tunit}_{\ast}\geq\Halt_{\am{\ket{\,},\ket{\,},M}}$.
\end{corollary}

\begin{proof}
By induction on $n$ and using Proposition~\ref{eq:soundness} we can show that 
$
\denot{M}^{\ell\vdash\tunit}_\ast(\qarray\qarray^\ast)
$
is greater or equal to 
$
\sum_{\am{\qarray'\!,\ell'\!,V}}\Red^{n}_{\am{\ell,\qarray,M},
\am{\qarray'\!,\ell'\!,V}}$. 
Then $\denot{M}^{\ell\vdash\tunit}_{\ast}(\qarray\qarray^\ast)
\geq\Halt_{\am{\qarray,\ell,M}}$ follows by taking the limit as 
$n\to\infty$,and
invoking the monotonicity of $\{\Red^{n}\}_n$.
\end{proof}

\subsection{The denotations of {\bf qlist} and {\bf teleport}}
\label{sect:examples}

\begin{example}\label{ex:term_sem}
  Recall the terms of Example~\ref{ex:term_type}. The web of
  $\denot{\tlist\qubit}$ is $\N$, while
  $\denot{\tlist\qubit}_n{=}(2^n,\{\id\})$. Note that
  $\pmatr(\denot{\tlist\qubit})$ is equivalent to 
  the D-completion of $\bigoplus_nP(\C^{2^n\times 2^n})$ where 
  the set
  $P(\C^{2^n\times 2^n})$ is
  the cone of $2^n\times2^n$ positive matrices.  The denotation of the
  term {\bf qlist} is a morphism in
  $\freecat[\ccpms](\qubit,\tlist\qubit)$, that is, a map sending a
  $2\times 2$ positive matrix onto $\bigoplus_nP(\C^{2^n\times 2^n})$.
  The program {\bf qlist} is defined using recursion: its semantics is
  the limit of the morphisms $f_n$ sending
  $(\begin{smallmatrix}a&b\\c&d\end{smallmatrix})$ to the 
  infinite sequence $ ({\bf
    0},\frac12e_1,\ldots\!,\frac1{2^n}e_n,{\bf 0},{\bf 0},\ldots)$ 
  where $e_i$ is the $2^i{\times}2^i$ positive matrix
\begin{center}
\vspace{-1ex}
\scalebox{0.6}{$\begin{pmatrix}
  a&0&\cdots&0&b
  \\
  0&0&\cdots&0&0
  \\[-1ex]
  \vdots&\vdots&\ddots&\vdots&\vdots
  \\[0.4ex]
  0&0&\cdots&0&0
  \\
  c&0&\cdots&0&d
\end{pmatrix}$}.
\end{center}
\vspace{-1ex}
This limit is the map sending 
$(\begin{smallmatrix}a&b\\c&d\end{smallmatrix})$ to the sequence of
infinitely increasing matrices
$
({\bf 0},\frac12e_1,\ldots,\frac1{2^n}e_n,\ldots)
$.
Note that the first element of the sequence is ${\bf 0}$, as the
program {\bf qlist} never return the empty list. Also note that all
the positive matrices in the sequence represent {\em entangled states
  of arbitrary sizes}. Our semantics is the first one to be able to
account for such a case: in~\cite{GoIquantum}, only fixed sizes
were allowed for entangled states.
\end{example}

\begin{example}\label{ex:telep-denot}
  We claim in the introduction that the model is expressive enough to
  describe entanglement at higher-order types. As we discuss in
  Example~\ref{ex:telep-term}, the encoding of the 
  quantum teleportation algorithm produces two entangled, mutually inverse
  functions: $f:\qubit\loli\bit\tensor\bit$ and
  $g:\bit\tensor\bit\loli\qubit$.
  
  The term $({\bf teleport}\,\punit)$ of type $(\qubit \loli
  \bit\otimes\bit)\tensor(\bit\otimes\bit \loli \qubit)$ is one
  instance of such a pair of functions. Its denotation is a finite sequence
  of $16$ square matrices of size $4\times 4$. Using a lexicographic
  convention, we can lay them out as in Fig.~\ref{tab:Axyzt}.
  Because of the convention, morally each row corresponds to an
  element of type $\bit\otimes\bit \loli \qubit$ whereas each column
  corresponds to an element of type $\qubit\loli
  \bit\otimes\bit$. Picking a row, i.e., a choice of two left-sided
  booleans, amounts to choosing the two booleans that will be passed to
  the function $g$. Picking a column, i.e., a choice of two right-sided
  booleans, amounts to deciding on the probabilistic result we get
  from the function $f$.  The intersection of a column and a row is
  therefore the representation of a map $\qubit\loli\qubit$. This map
  is a description of a possible path in the control flow of the
  algorithm. 
  
  The matrices on the diagonal correspond to a run of the algorithm
  as it was intended: applying $g$ to the result of $f$. Since they
  are supposed to be the identity on $\qubit$, we can therefore deduce
  that the matrices $A_{00,00}$, $A_{01,01}$, $A_{10,10}$ and
  $A_{11,11}$ are all equal to $\left(\begin{smallmatrix}
      1&0&0&1\\
      0&0&0&0\\
      0&0&0&0\\
      1&0&0&1
  \end{smallmatrix}\right).
  $
  Since this matrix cannot be written as the tensor of two $2\times2$
  matrices, we conclude that the denotation $A$ of $({\bf
    teleport}\,\punit)$ is indeed entangled.
  
  We can compute the other matrices $A_{xy,zt}$ using the same
  argument: in general, $A_{xy,zt}$ is a composition of $f$ and $g$,
  except that instead of applying $g$ to $(x,y)$, we apply it to 
  $(z,t)$. We therefore get a function $\qubit\to\qubit$ constructed
  out of the $U_{--}$ that might (if $xy=zt$) or might not be the
  identity. In general, the matrix $A_{xy,zt}$ is the denotation of
  the unitary $U_{zt}U_{xy}^*$. The denotation $A$ is given in full
  detail in Table~\ref{tab:Axyzt}.
\end{example}

\begin{remark}
  Example~\ref{ex:telep-denot} is a good illustration of what we
  claimed in the introduction: the model  reflects the
  juxtaposition of quantum and classical structures, even at
  higher-order types. Here, the control-flow is handled by
  the biproduct structure, and the quantum part of the algorithm is
  split across the list of $4{\times}4$~matrices.
\end{remark}

\begin{table*}
  \[
  \begin{array}{
      r@{{}={}}c@{\quad}c@{{}={}}c@{\quad}c@{{}={}}c@{\quad}c@{{}={}}cc}
    A \quad=\quad\frac14\bigg(\quad
    A_{00,00}&
    \left(\begin{smallmatrix}
        1&0&0&1\\
        0&0&0&0\\
        0&0&0&0\\
        1&0&0&1
      \end{smallmatrix}\right),
    &
    A_{00,01}&
    \left(\begin{smallmatrix}
        0&0&0&0\\
        0&1&1&0\\
        0&1&1&0\\
        0&0&0&0
      \end{smallmatrix}\right),
    &
    A_{00,10}&
    \left(\begin{smallmatrix}
        1&0&0&\textrm{-}1\\
        0&0&0&0\\
        0&0&0&0\\
        \textrm{-}1&0&0&1
      \end{smallmatrix}\right),
    &
    A_{00,11}&
    \left(\begin{smallmatrix}
        0&0&0&0\\
        0&1&\textrm{-}1&0\\
        0&\textrm{-}1&1&0\\
        0&0&0&0
      \end{smallmatrix}\right),
    \\
    A_{01,00}&
    \left(\begin{smallmatrix}
        0&0&0&0\\
        0&1&1&0\\
        0&1&1&0\\
        0&0&0&0
      \end{smallmatrix}\right),
    &
    A_{01,01}&
    \left(\begin{smallmatrix}
        1&0&0&1\\
        0&0&0&0\\
        0&0&0&0\\
        1&0&0&1
      \end{smallmatrix}\right),
    &
    A_{01,10}&
    \left(\begin{smallmatrix}
        0&0&0&0\\
        0&1&\textrm{-}1&0\\
        0&\textrm{-}1&1&0\\
        0&0&0&0
      \end{smallmatrix}\right),
    &
    A_{01,11}&
    \left(\begin{smallmatrix}
        1&0&0&\textrm{-}1\\
        0&0&0&0\\
        0&0&0&0\\
        \textrm{-}1&0&0&1
      \end{smallmatrix}\right),
    \\
    A_{10,00}&
    \left(\begin{smallmatrix}
        1&0&0&\textrm{-}1\\
        0&0&0&0\\
        0&0&0&0\\
        \textrm{-}1&0&0&1
      \end{smallmatrix}\right),
    &
    A_{10,01}&
    \left(\begin{smallmatrix}
        0&0&0&0\\
        0&1&\textrm{-}1&0\\
        0&\textrm{-}1&1&0\\
        0&0&0&0
      \end{smallmatrix}\right),
    &
    A_{10,10}&
    \left(\begin{smallmatrix}
        1&0&0&1\\
        0&0&0&0\\
        0&0&0&0\\
        1&0&0&1
      \end{smallmatrix}\right),
    &
    A_{10,11}&
    \left(\begin{smallmatrix}
        0&0&0&0\\
        0&1&1&0\\
        0&1&1&0\\
        0&0&0&0
      \end{smallmatrix}\right),
    \\
    A_{11,00}&
    \left(\begin{smallmatrix}
        0&0&0&0\\
        0&1&\textrm{-}1&0\\
        0&\textrm{-}1&1&0\\
        0&0&0&0
      \end{smallmatrix}\right),
    &
    A_{11,01}&
    \left(\begin{smallmatrix}
        1&0&0&\textrm{-}1\\
        0&0&0&0\\
        0&0&0&0\\
        \textrm{-}1&0&0&1
      \end{smallmatrix}\right),
    &
    A_{11,10}&
    \left(\begin{smallmatrix}
        0&0&0&0\\
        0&1&1&0\\
        0&1&1&0\\
        0&0&0&0
      \end{smallmatrix}\right),
    &
    A_{11,11}&
    \left(\begin{smallmatrix}
        1&0&0&1\\
        0&0&0&0\\
        0&0&0&0\\
        1&0&0&1
      \end{smallmatrix}\right)&\quad\bigg).
  \end{array}
  \]
  \caption{\footnotesize The denotation of the quantum
    teleportation algorithm.}
  \label{tab:Axyzt}
\end{table*}

\section{Adequacy}\label{sec:adequacy}

In the following, we prove the adequacy of $\freecat[\ccpms]$
(Theorem~\ref{th:adequacy}). This amounts to achieving the converse
inequality of Corollary~\ref{cor:invariance_interpretation}. The proof
uses a syntactic approach, following~\cite{GoIquantum}. We introduce a
bounded {\tt letrec}$^n$, which can be unfolded at most $n$ times. On
the one hand, the language allowing only bounded {\tt letrec} is
strongly normalizing (Lemma~\ref{lem:SN-ext}), hence the adequacy for
it can be easily achieved by induction on the longest reduction
sequence of a term (Corollary~\ref{cor:finit_adequacy}). On the other
hand, the unbounded {\tt letrec} can be expressed as the supremum of
its bounded approximants, both semantically
(Lemma~\ref{lem:finitary0}) and syntactically
(Lemma~\ref{lem:finitary2}). We then conclude the adequacy for the
whole quantum lambda calculus by continuity.

\begin{definition}\label{def:letrecn}
  Let us extend the grammar of terms (Table~\ref{table:terms_grammar})
  by adding: (i) a new term $\Omega^A$; (ii) a family of new term
  constructs $\letrecn{n}{f^{A\loli B}}{x}{M}{N}$ indexed by natural
  numbers $n\geq 0$.
    
  The typing rules for these new constructs are
  \[
  \infer{
    \bang\Delta\entail\Omega^A:A
  }{}
  \qquad
  \infer{
    \bang\Delta,\Gamma\entail\letrecn{n}{f^{A\loli B}}{x}{M}{N}:C
  }{
    \begin{array}{l}
    \bang\Delta,f:\bang{(A\loli B)},x:A
    \entail
    M:B
    \\
    \bang\Delta,\Gamma,f:\bang{(A\loli B)}
    \entail
    N:C
    \end{array}
  }
  \]
  Their denotations are given, respectively, by the map $\bf 0$ and
  the family of maps
  \[
  \oc\Delta\tensor\Gamma
  \xrightarrow{\contr}
  \oc\Delta\tensor\Gamma\tensor\oc\Delta
  \xrightarrow{\!\!\id\tensor(\dig;\bierman;\oc{(\Lambda\phi)})^n\!\!}
  \oc\Delta\tensor\Gamma\tensor\oc{(A\loli B)}
  \xrightarrow{\psi}
  C,
  \]
  where $\phi\in\freecat[\ccpms](\oc\Delta\otimes\oc(A\multimap
  B)\otimes A,B)$ and
  $\psi\in\freecat[\ccpms](\oc\Delta\otimes\Gamma\otimes\oc(A\multimap
  B),C)$ are the denotations of the premises and
  $(\dig;\bierman;\oc{(\Lambda\phi)})^n
  \in\freecat[\ccpms](\oc\Delta,\oc(A\multimap
  B))$ is defined in a similar fashion as in
  Table~\ref{table:denotation_rule}.
  
  The reduction rules are updated as follows.
  \[
  \begin{array}{@{}l@{}}
    \am{q,\ell,\letrecn{0}{f^{A\loli B}}{x}{M}{N}} \xredto[1]
    \am{q,\ell,N\{(\lambda x^A.\Omega^B)/f\}}
    \\[1ex]
    \am{q,\ell,\letrecn{n+1}{f^{A\loli B}}{x}{M}{N}} \\
    ~\qquad~\xredto[1]
    \am{q,\ell,N\{(\lambda x^A.\letrecn{n}{f^{A\loli B}}{x}{M}{M})/f\}}.
  \end{array}
  \]
\end{definition}

The additions to the language do not modify the properties of the
language: subject reduction (Proposition~\ref{prop:subject_reduction})
and totality (Lemma~\ref{lemma:totality}) hold as they are stated,
while type safety (Proposition~\ref{prop:safety}) and soundness
(Proposition~\ref{prop:semantic_invariance}) are satisfied, with the
proviso of considering the set of normal forms to consist of the set
of values \emph{and} the set of terms containing $\Omega$ in
evaluating position.
\begin{definition}\label{def:finitary}
  A term is called \define{finitary} when it does not contain any
  occurrence of the un-indexed {\tt letrec} construct. It can however
  contain $\Omega$ and any of the indexed {\tt letrec}${}^n$. We call a
  closure \define{finitary} when its term is finitary.
\end{definition}

\begin{lemma}[Strong normalization]\label{lem:SN-ext}
  If $\am{\qarray_1,\qlist_1,M_1}$ is finitary and typable, then every
  reduction sequence of the form
  $
  \am{\qarray_1,\qlist_1,M_1}\xredto[p_1]
  \am{\qarray_2,\qlist_2,M_2}\xredto[p_2]
  \am{\qarray_3,\qlist_3,M_3}\xredto[p_3]
  \cdots
  $
  is finite.
\end{lemma}
\begin{proof}[Proof (Sketch)]
  We reduce the finitary quantum lambda calculus to a simply typed
  non-deterministic language without quantum states, for which a
  standard proof technique can be used. The terms of this language are
  the terms of the extended quantum lambda calculus, minus the {\tt
    letrec} construct. The operational semantics is obtained from
  Table~\ref{table:reduction} and the rules for {\tt letrec}$^n$ by
  replacing closures with the respective terms and the rules of
  Table~\ref{subtable:reduction_quantum} by dummy reduction rules:
  like $ U(\bullet\otimes\dots\otimes\bullet) \redto
  \bullet\otimes\dots\otimes\bullet$, or $ \new~\ffalse \redto
  \bullet$. The symbol $\bullet$ denotes a distinct term variable,
  which, by convention, it is never bound by an abstraction. Clearly,
  the strong normalization of this language implies that of the
  finitary quantum lambda calculus.
\end{proof}

\begin{corollary}[Finitary adequacy]\label{cor:finit_adequacy}
  Let $M$ be a closed finitary term of unit type. Then $
  \denot{M}^{\vdash\tunit}_\ast=\Halt_{\am{\ket{\,},\ket{\,},M}}.  $
\end{corollary}
\begin{proof}[Proof (Sketch)]
  We prove that, for any total finitary quantum closure of unit type
  $\am{\qarray,\qlist,M}$ we have $\denot{M}^{\qlist\vdash
    1}(\qarray\qarray^\ast)=\Halt_{\am{\qarray,\qlist,M}}$. In fact,
  by Lemma~\ref{lem:SN-ext}, there exists $m\in\N$ such that
  $\Halt_{\am{\qarray,\qlist,M}}= \sum_{\am{\qarray',\qlist',V}}
  \Red^m_{\am{\qarray,\!\qlist,\!M},\am{\qarray'\!,\!\qlist'\!,\!V}}$. We
  conclude by induction on $m$.
\end{proof}

\begin{definition}\label{def:approx}
  Let $\apprv$ be a relation between finitary terms and general terms
  defined as the smallest congruence relation on terms satisfying, for
  every $M\apprv M'$ and $N\apprv N'$:
  \begin{align*}
  N\{(\lambda x^A.\Omega^B)/f\}&\apprv(\letrec{f}{x}{M'}{N'}),\\
  (\letrecn{n}{f}{x}{M}{N})&\apprv(\letrec{f}{x}{M'}{N'}).
  \end{align*}
\end{definition}
\begin{lemma}\label{lem:finitary0}
If $\Gamma\vdash M:A$, then $
	\denot{M}^{\Gamma\vdash A}\!\!\!\!=\dirsup_{\!\!
		\substack{
		M'\apprv M\\
		M'\text{ finitary}\!\!\!\!\!
		}
	}\denot{M'}^{\Gamma\vdash A}$.\qed
\end{lemma}

\begin{lemma}\label{lem:finitary2}
  If $M\apprv M'$, then
  $\Halt_{\am{\qarray,\qlist,M}}\leq\Halt_{\am{\qarray,\qlist,M'}}$.
\end{lemma}
\begin{proof}[Proof (Sketch)]
  By induction on $n$, one proves the inequality:
  $\sum_{\am{q',\ell',V}}\Red^n_{\am{q,\ell,M},\am{q',\ell',V}} \leq
  \sum_{\am{q',\ell',V}}\Red^n_{\am{q,\ell,M'},\am{q',\ell',V}}$, from
  which the statement follows trivially.
\end{proof}

\begin{theorem}\label{th:adequacy}
Let $M$ be a program, i.e., a closed term of unit type. Then 
$
\denot{M}^{\vdash\tunit}_\ast=\Halt_{\am{\ket{\,},\ket{\,},M}}.\
$
\end{theorem}
\begin{proof}
  By Corollary~\ref{cor:invariance_interpretation} we have
  $\denot{M}^{\vdash\tunit}_\ast\geq\Halt_{\am{\ket{\,},\ket{\,},M}}$.
  Conversely,
  by Lemma~\ref{lem:finitary0},
  $\denot{M}^{\vdash\tunit}_\ast=\dirsup_{M'\apprv
    M}\denot{M'}^{\vdash\tunit}_\ast$, which is equal to
  $\dirsup_{M'\apprv M}\Halt_{\am{\ket{\,},\ket{\,},M'}}$ by Corollary
  \ref{cor:finit_adequacy}, which is less or equal to
  $\Halt_{\am{\ket{\,},\ket{\,},M}}$ by Lemma~\ref{lem:finitary2}.
\end{proof}

\section{Structure of the sets of representable elements}
\label{subsect:discussion}

We conclude this paper with an analysis of some of the properties of
the denotation of terms. 
Recall that a morphism in $\freecat[\ccpms]$ is an indexed family
of either
completely positive maps, or infinite elements added during D-completion.
We show that (1) all types have a non-zero inhabitant; (2) provided
that the term constant $U$ ranges over arbitrary unitary
matrices, the representable elements of a given homset form a convex
set including ${\bf 0}$; and (3) infinite elements are not part of
any representable map.

We first need two auxiliary definitions.

\newcommand{\consume}[1][]{\overline\omega_{#1}}
\newcommand{\generate}[1][]{\omega_{#1}}

\begin{table*}
  \[
  \begin{array}{@{}l@{{}={}}ll@{{}={}}l@{}}
    \generate[\qubit]
    & 
    \lambda\punit.\new\,\ffalse
    &
    \consume[\qubit]
    &
    \lambda x^\qubit.\iftermx{\meas\,x}{\punit}{\punit}
    \\
    \generate[A\loli B]
    &
    \lambda\punit.\lambda x^A.\letunitterm{(\consume[A]\,x)}{
      (\generate[B]\,\punit)}
    &
    \consume[A\loli B]
    &
    \lambda f^{A\loli B}.\consume[B]\,(f\,(\generate[A]\,\punit))
    \\
    \generate[\oc{(A\loli B)}]
    &
    \lambda\punit.\lambda x^A.(\generate[A\loli B]\,\punit)\,x
    &
    \consume[\oc{(A\loli B)}]
    &
    \mu g f^{\oc{(A\loli B)}}.\iftermx{{\bf c}}{\punit}{
      (\consume[A\loli B]\,f);(g\,f)}
    \\
    \generate[\tunit]
    &
    \lambda\punit.\punit
    &
    \consume[\tunit]
    &
    \lambda\punit.\punit
    \\
    \generate[A\tensor B]
    &
    \lambda\punit.(\generate[A]\,\punit)\tensor(\generate[B]\,\punit)
    &
    \consume[A\tensor B]
    &
    \lambda x^{A\tensor B}.\lettensterm{z_1}{z_2}{x}{
      \letunitterm{(\consume[A]\,z_1)}{(\consume[B]\,z_2)}}
    \\
    \generate[A\oplus B]
    &
    \lambda\punit.\iftermx{{\bf c}}{(\generate[A]\,\punit)}{
      (\generate[B]\,\punit)}
    &
    \consume[A\oplus B]
    &
    \lambda x^{A\oplus B}.\match{x}{z_1^A}{\consume[A]\,z_1}{z_2^B}{
      \consume[B]\,z_2}
    \\
    \generate[\tlist{A}]
    &
    \mu f
    \punit.\iftermx{{\bf c}}{(\punit)}{\cons{(\generate[A]\,\punit)}{
        (f\,\punit)}}
    &
    \consume[\tlist{A}]
    &\begin{array}[t]{@{}l@{}}
      \mu f
      x^{\tlist{A}}.{\tt match}\,\splitlist\,x\,{\tt with}~
      \\
      \hspace{-.43in}(~z_1^\tunit~:~z_1\bor z_2^{A\tensor\tlist{A}}~:~
      {\tt let}~{y_1}\tensor{y_2}={z_2}~{\tt in}~{
        \letunitterm{(\consume[A]\,y_1)}{(f\,y_2)}})
   \end{array}
 \end{array}
  \]
  \caption{\footnotesize Two mutually recursive families of terms}
  \label{tab:consume-generate}
\end{table*}

\begin{definition}
  We define two type-indexed families of terms $\consume[A]$ and
  $\generate[A]$ by mutual induction in
  Table~\ref{tab:consume-generate}. The term ${\bf c}$ represents the
  fair coin toss $\meas\,(H\,(\new\,\ffalse))$ (recall
  Example~\ref{ex:cointoss}) and the notation $\mu f x.M$ stands for
  $\letrec{f}{x}{M}{f}$.
\end{definition}

\begin{lemma}\label{lem:non-zero-image}
  For all types $A$, we have $\entail\generate[A]:\tunit\loli A$ and
  $\entail\consume[A]:A\loli\tunit$. Moreover, the morphisms
  $\denot{\generate[A]}^{\entail\tunit\loli A}$ and
  $\denot{\consume[A]}^{\entail A\loli\tunit}$, seen as indexed families,
  do not contain the zero map.\qed
\end{lemma}

\begin{corollary}
  All types are inhabited by at least one closed value of non-null
  denotation.
\end{corollary}

\begin{proof}
  Immediate with Lemma~\ref{lem:non-zero-image}: for a given type $A$,
  choose the term $(\generate[A]\,\punit)$.
\end{proof}

\begin{proposition}
  Given a type $A$ and a context $\Gamma$, the denotations
  $\denot{M}^{\Gamma\entail A}$ of valid typing judgements $\Gamma\entail
  M:A$ form a convex set including ${\bf 0}$.
\end{proposition}

\begin{proof}
  Suppose that $\Gamma$ is $x_1:A_1,\ldots,x_n:A_n$.  
  A term $M$ mapping to $\bf 0$ is 
  $
  (\consume[A_1] x_1;\ldots;\consume[A_n] x_n;{\bf\Omega})
  $
  where the term ${\bf\Omega}$ is a shortcut for 
  $\letrec{f}{x}{f\,x}{f\,\punit}$, of denotation
  $\bf 0$. 
  
  Now, suppose that $f=\denot{M_1}^{\Gamma\entail A}$ and
  $g=\denot{M_2}^{\Gamma\entail A}$, and choose two non-negative real
  numbers $\rho_1$, $\rho_2$ such that $\rho_1+\rho_2 = 1$.
  There exists an angle $\phi$ such that $(\cos \phi)^2=\rho_1$ and
  that $(\sin \phi)^2=\rho_2$. 
  As the term constants $U$ range over arbitrary unitaries,
  the unitary matrix $V_\phi=
  (\begin{smallmatrix}\cos\phi&-\sin\phi\\\sin\phi&\cos
    \phi\end{smallmatrix})$
  is representable in the quantum lambda calculus. The term
  $
  {\bf c'} = \meas\,(V_\phi\,(\new\,\ffalse))
  $
  has denotation $(\rho_1,\rho_2)$. We then conclude that
  the term 
  $\iftermx{{\bf c'}}{M_1}{M_2}$
  has denotation $\rho_1f+\rho_2g$.
\end{proof}

\begin{proposition}
  If $\Gamma\entail M:A$ is valid, then no infinite element is part of
  the denotation $\denot{M}^{\Gamma\entail A}$ of $M$.
\end{proposition}

\begin{proof}
  Suppose that one of the infinite elements of the D-comple\-tion were
  to be found in the interpretation of $x_1:A_1,\ldots,x_n:A_n\entail
  M:A$. Then the closed term
  \[
  (\lambda x_1\ldots x_n.\consume[A]\,M)
  (\generate[A_1]\punit)\ldots(\generate[A_n]\punit)
  \]
  of type $\tunit$ has infinite denotation, contradicting
  Theorem~\ref{th:adequacy}.
\end{proof}

This last proposition indicates that infinite elements introduced
during the D-completion are really an
artifact only needed for the categorical construction. The
representable elements in the model are only built out of families of
completely positive maps.

\section{Conclusion}

We presented a higher-order lambda calculus for quantum computation
featuring classical and quantum data, duplication, recursion, and an
infinite parametric type for lists. We then answered a long-standing
open question: the description of a model for the full quantum lambda
calculus. The model we propose is a free construction based on the
known model of completely positive maps, but nevertheless has a
concrete presentation.

One thing that this model explains and illustrates is the distinction
between the quantum and classical parts of the language. The quantum
part is described by completely positive maps (finite dimension), 
whereas the classical
control is given by the Lafont category (i.e., linear logic). The
model demonstrates that the two ``universes'' work well together, but
also -- surprisingly -- that they do not mix too much, even at higher
order types (we always have an {\em infinite} list of {\em finite}
dimensional CPMs). The control flow is completely handled by the
biproduct completion, and not by the CPM structure.
The adequacy result, moreover, validates that the model is a ``good''
representation of the language. 

One should also note that the product and the coproduct coincide in
our model. For example, the model has morphisms that correspond to a
program returning true with probability $1$ and false with probability
$1$.  We would like to point out that our interpretation is not
{\em surjective}. For example, there are also morphisms in the model
corresponding to ``probability 2''.  (Incidentally, adding terms with
such behavior makes it possible to build a term whose denotation is
$\infty$ -- so the fact that this provably does not happen somehow
captures the sanity of the model).  Interpretations in denotational
models are often not surjective. In fact, it is an open problem to
give a non-syntactic characterization of the image of our
interpretation. Similarly, the problem of full-abstraction is still
open.

\bibliographystyle{abbrv}
\bibliography{bib}

\newcommand{\online}[1]{Available at \url{#1}}
\begin{thebibliography}{10}

\bibitem{danosehrhard}
V.~Danos and T.~Ehrhard.
\newblock Probabilistic coherence spaces as a model of higher-order
  probabilistic computation.
\newblock {\em Inform. Comput.}, 2011.

\bibitem{finsp}
T.~Ehrhard.
\newblock Finiteness spaces.
\newblock {\em MSCS}, 15(4):615--646, 2005.

\bibitem{ll}
J.-Y. Girard.
\newblock Linear logic.
\newblock {\em Th. Comp. Sc.}, 50:1--102, 1987.

\bibitem{Girard88c}
J.-Y. Girard.
\newblock Normal functors, power series and lambda-calculus.
\newblock {\em Ann. Pure Appl. Logic}, 37(2):129--177, 1988.

\bibitem{Girard99coherentbanach}
J.-Y. Girard.
\newblock Coherent {Banach} spaces: a continuous denotational semantics.
\newblock {\em Theoretical Computer Science}, 227:297, 1999.

\bibitem{GoIquantum}
I.~Hasuo and N.~Hoshino.
\newblock Semantics of higher-order quantum computation via geometry of
  interaction.
\newblock In {\em Proceedings of LICS}, pages 237--246, 2011.

\bibitem{Keimel2009}
K.~Keimel and J.~D. Lawson.
\newblock {D}-completions and the {\em d}-topology.
\newblock {\em Annals of Pure and Applied Logic}, 159(3):292 -- 306, 2009.

\bibitem{knill}
E.~H. Knill.
\newblock Conventions for quantum pseudocode.
\newblock Technical Report LAUR-96-2724, Los Alamos National Laboratory, 1996.

\bibitem{lafont:these}
Y.~Lafont.
\newblock {\em Logiques, cat\'egories et machines}.
\newblock PhD thesis, Universit\'e Paris 7, 1988.

\bibitem{LagoMZ11}
U.~D. Lago, A.~Masini, and M.~Zorzi.
\newblock Confluence results for a quantum lambda calculus with measurements.
\newblock {\em Electr. Notes Theor. Comput. Sci.}, 270(2):251--261, 2011.

\bibitem{LairdMM12}
J.~Laird, G.~Manzonetto, and G.~McCusker.
\newblock Constructing differential categories and deconstructing categories of
  games.
\newblock {\em Information and Computation}, 222:247--264, 2013.

\bibitem{LairdMMP13}
J.~Laird, G.~McCusker, G.~Manzonetto, and M.~Pagani.
\newblock Weighted relational models of typed lambda-calculi.
\newblock In {\em {LICS'13}}, 2013.

\bibitem{malherbe2010}
O.~Malherbe.
\newblock {\em Categorical models of computation: partially traced categories
  and presheaf models of quantum computation}.
\newblock PhD thesis, University of Ottawa, 2010.

\bibitem{melliespanorama}
P.-A. Melli\`es.
\newblock Categorical semantics of linear logic.
\newblock {\em Panoramas et Synth\`eses}, 12, 2009.

\bibitem{MelliesTT09}
P.-A. Melli{\`e}s, N.~Tabareau, and C.~Tasson.
\newblock An explicit formula for the free exponential modality of linear
  logic.
\newblock In {\em ICALP'09 (2)}, pages 247--260, 2009.

\bibitem{nielsen02quantum}
M.~A. Nielsen and I.~L. Chuang.
\newblock {\em Quantum Computation and Quantum Information}.
\newblock Cambridge University Press, 2002.

\bibitem{Selinger04}
P.~Selinger.
\newblock Towards a quantum programming language.
\newblock {\em Mathematical Structures in Computer Science}, 14(4):527--586,
  2004.

\bibitem{Selinger04b}
P.~Selinger.
\newblock Towards a semantics for higher-order quantum computation.
\newblock In {\em QPL'04}, TUCS General Publication No 33, pages 127--143,
  2004.

\bibitem{selinger06lambda}
P.~Selinger and B.~Valiron.
\newblock A lambda calculus for quantum computation with classical control.
\newblock {\em Mathematical Structures in Computer Science}, 16(3):527--552,
  2006.

\bibitem{valiron06fully}
P.~Selinger and B.~Valiron.
\newblock On a fully abstract model for a quantum linear functional language.
\newblock In {\em QPL'06}, 2008.

\bibitem{SV09}
P.~Selinger and B.~Valiron.
\newblock Quantum lambda calculus.
\newblock In S.~Gay and I.~Mackie, editors, {\em Semantic Techniques in Quantum
  Computation}, chapter~9, pages 135--172. Cambridge University Press, 2009.

\bibitem{valiron08phd}
B.~Valiron.
\newblock {\em Semantics for a higher-order functional programming language for
  quantum computation}.
\newblock PhD thesis, University of Ottawa, 2008.

\bibitem{ZhaoFan2010}
D.~Zhao and T.~Fan.
\newblock Dcpo-completion of posets.
\newblock {\em Th. Comp. Sc.}, 411(22--24):2167--–2173, 2010.

\end{thebibliography}

\end{document}